\documentclass[10pt,a4paper]{article}

\newcommand{\macrospath}{latex-macros}

\usepackage{ifthen}

\newboolean{talk}
\setboolean{talk}{false}
\newboolean{paper}
\setboolean{paper}{false}

\newboolean{IEEEstyle}
\setboolean{IEEEstyle}{false}
\newboolean{lipicsstyle}
\setboolean{lipicsstyle}{false}
\newboolean{eptcsstyle}
\setboolean{eptcsstyle}{false}
\newboolean{sigplanstyle}
\setboolean{sigplanstyle}{false}

\newboolean{needstheorems}
\setboolean{needstheorems}{false}
\newboolean{withimages}
\setboolean{withimages}{false}
\newboolean{withproofs}
\setboolean{withproofs}{true}

\newboolean{french}
\setboolean{french}{false}

\ifthenelse{\boolean{talk}}{}{\usepackage[utf8]{inputenc}}
\ifthenelse{\boolean{french}}{\usepackage[french]{babel}}{
  \ifthenelse{\boolean{lipicsstyle}}{}{\usepackage[english]{babel}}}
\usepackage{amssymb}
\usepackage{amsmath}
\usepackage{graphicx}
\usepackage{bussproofs}
\usepackage{fancybox}
\usepackage{cmll}
\usepackage{stmaryrd}
\usepackage{multirow}
\ifthenelse{\boolean{talk}}{\usepackage{color}}{
  \ifthenelse{\boolean{lipicsstyle}}{\usepackage{color}}{\usepackage[usenames]{xcolor}}}
\ifthenelse{\boolean{IEEEstyle}}{
\setboolean{needstheorems}{true}}{}

\ifthenelse{\boolean{eptcsstyle}}{
\setboolean{needstheorems}{true}}{}

\ifthenelse{\boolean{sigplanstyle}}{
\setboolean{needstheorems}{true}}{}

\ifthenelse{\boolean{needstheorems}}{

\usepackage{amsthm}

    \newtheorem{theorem}{Theorem}[section]
    \newtheorem{lemma}[theorem]{Lemma}
    \newtheorem{corollary}[theorem]{Corollary}
    \newtheorem{proposition}[theorem]{Proposition}
    \newtheorem{definition}[theorem]{Definition}
    \newtheorem{remark}[theorem]{Remark}
}{}

\newcommand{\myproof}[1]{
\ifthenelse{\boolean{withproofs}}{#1}{}}

\newcommand{\withproofs}[1]{
\ifthenelse{\boolean{withproofs}}{#1}{}}

\newcommand{\withoutproofs}[1]{
\ifthenelse{\boolean{withproofs}}{}{#1}}

\newcommand{\tm}{t}
\newcommand{\tmtwo}{u}
\newcommand{\tmthree}{r}
\newcommand{\tmfour}{p}
\newcommand{\tmfive}{s}

\newcommand{\var}{x}
\newcommand{\vartwo}{y}
\newcommand{\varthree}{z}

\newcommand{\rootRew}[1]{\mapsto_{#1}}
\newcommand{\Rew}[1]{\rightarrow_{#1}}

\renewcommand{\to}{\Rew{}}
\newcommand{\rtob}{\rootRew{\beta}}
\newcommand{\tob}{\Rew{\beta}}

\newcommand{\lssym}{{\tt ls}}

\newcommand{\db}{{\tt dB}}

\newcommand{\ctxholep}[1]{[#1]}
\newcommand{\ctxhole}{\ctxholep{\cdot}}

\newcommand{\ctx}{C}
\newcommand{\ctxtwo}{D}
\newcommand{\ctxthree}{E}
\newcommand{\ctxfour}{F}
\newcommand{\ctxp}[1]{\ctx\ctxholep{#1}}

\newcommand{\apctx}{A}
\newcommand{\apctxtwo}{B}

\newcommand{\nbvctxtwo}[1]{\nbvctxtwo{#1}}

\newcommand{\sctx}{L}
\newcommand{\sctxtwo}{L'}
\newcommand{\sctxthree}{L''}

\newcommand{\sctxp}[1]{\sctx\ctxholep{#1}}

\newcommand{\sctxtwop}[1]{\sctxtwo\ctxholep{#1}}
\newcommand{\sctxthreep}[1]{\sctxthree\ctxholep{#1}}

\newcommand{\defeq}{:=}
\newcommand{\grameq}{::=}

\newcommand{\isub}[2]{\{#1/#2\}}
\newcommand{\esub}[2]{[#1/#2]}

\newcommand{\rtodb}{\rootRew{\db}}

\newcommand{\todbp}[1]{\Rew{\db#1}}
\newcommand{\todb}{\todbp{}}%{\rightsquigarrow}

\newcommand{\rtols}{\rootRew{\lssym}}

\newcommand{\tols}{\Rew{\lssym}}

\newcommand{\llbrace}{\{ \kern -0.27em \vert}
\newcommand{\rrbrace}{\vert \kern -0.27em \}}

\renewcommand{\l}{\lambda}
\newcommand{\ie}{{\em i.e.}}
\newcommand{\eg}{{\em e.g.}}
\newcommand{\ih}{{\em i.h.}}
\newcommand{\fv}[1]{{\tt fv}(#1)}

\newcommand{\deff}[1]{\textbf{#1}}

\newcommand{\ben}[1]{{\color{red} {#1}}}

\newcommand{\cben}[2]{{\color{red} {#2}}}

\newcommand{\ignore}[1]{}

\newcommand{\myinput}[1]{\ifthenelse{\boolean{withimages}}{\input{#1}}{}}

\newcommand{\mellies}{{Melli{\`e}s}}
\newcommand{\levy}{{L{\'e}vy}}

\newcommand{\linlogic}{linear logic}

\newcommand{\set}[1]{\{#1\}}

\newcommand{\nat}{\mathbb{N}}

\newcommand{\setone}{S}
\newcommand{\settwo}{R}

\newcommand{\pns}{proof nets}

\newcommand{\size}[1]{|#1|}

\newenvironment{varitemize}
{
\begin{list}{\labelitemi}
{\setlength{\itemsep}{0pt}
 \setlength{\topsep}{0pt}
 \setlength{\parsep}{0pt}
 \setlength{\partopsep}{0pt}
 \setlength{\leftmargin}{15pt}
 \setlength{\rightmargin}{0pt}
 \setlength{\itemindent}{0pt}
 \setlength{\labelsep}{5pt}
 \setlength{\labelwidth}{10pt}
}}
{
 \end{list} 
}

\newcounter{numberone}

\newenvironment{varenumerate}
{
\begin{list}{\arabic{numberone}.}
{
  \usecounter{numberone}
  \setlength{\itemsep}{0pt}
  \setlength{\topsep}{0pt}
  \setlength{\parsep}{0pt}
  \setlength{\partopsep}{0pt}
  \setlength{\leftmargin}{15pt}
  \setlength{\rightmargin}{0pt}
  \setlength{\itemindent}{0pt}
  \setlength{\labelsep}{5pt}
  \setlength{\labelwidth}{15pt}
}}
{
\end{list} 
}

\newcounter{numbertwo}

\usepackage{tikz}
\usetikzlibrary{matrix,arrows}
\usetikzlibrary{decorations.shapes}
\usetikzlibrary{decorations.text}
\usetikzlibrary{decorations.pathmorphing}
\usetikzlibrary{decorations.markings}

\tikzset{
node distance=1.3cm, auto,
every node/.style={font=\tiny },
ocenter/.style={baseline={([yshift=-.5ex, xshift=-.5ex]current bounding box)}},  
labelBeginAbove/.style={postaction={decorate,decoration={markings,mark=at position 0 with {\node[inner sep= 0.6pt, above=1pt]{\tiny #1};}} } },
labelBeginBelow/.style={postaction={decorate,decoration={markings,mark=at position 0 with {\node[inner sep= 0.6pt, below=1pt]{\tiny #1};}}}},
labelEndAbove/.style={postaction={decorate,decoration={markings,mark=at position 1 with {\node[inner sep= 0.6pt, above=1pt]{\tiny #1};}}}},
labelEndBelow/.style={postaction={decorate,decoration={markings,mark=at position 1 with {\node[inner sep= 0.6pt, below=1pt]{\tiny #1};}}}},
labelEndRight/.style={postaction={decorate,decoration={markings,mark=at position 1 with {\node[inner sep= 0.6pt, right=1pt]{\tiny #1};}}}},
labelEndLeft/.style={postaction={decorate,decoration={markings,mark=at position 1 with {\node[inner sep= 0.6pt, left=1pt]{\tiny #1};}}}}
}

\newcommand{\nodeHorDist}{2cm}
\newcommand{\nodeVerDist}{1cm}

\newcommand{\commDiagramRed}[8]{
\begin{tikzpicture}[ocenter]
	\node (s) {\normalsize #1};
  \node at (s.center)  [right=1.7*\nodeHorDist](s1){\normalsize #2};
  \node at (s.center)  [below=\nodeVerDist](s2) {\normalsize #3};
  \node at (s1|-s2) (t) {\normalsize #4};
  
  \draw[->] (s) to node {#5} (s1);
  \draw[->] (s2) to node {#6} (t);
  \draw[->] (s) to node {#7} (s2);
  \draw[->] (s1) to node {#8} (t);
\end{tikzpicture} 
}

\renewcommand{\ctxholep}[1]{\langle #1\rangle}
\newcommand{\ctxtwop}[1]{\ctxtwo\ctxholep{#1}}
\newcommand{\ctxthreep}[1]{\ctxthree\ctxholep{#1}}
\newcommand{\ctxfourp}[1]{\ctxfour\ctxholep{#1}}

\newcommand{\apctxp}[1]{\apctx\ctxholep{#1}}
\newcommand{\apctxtwop}[1]{\apctxtwo\ctxholep{#1}}

\newcommand{\esmeas}[1]{|#1|_{[\cdot]}}

\newcommand{\unfsym}{\rotatebox[origin=c]{-90}{$\rightarrow$}}
\newcommand{\unf}[1]{#1\unfsym}
\newcommand{\relunf}[2]{\unf{#1}_{#2}}

\newcommand{\opt}{useful}

\newcommand{\deriv}{\rho}
\newcommand{\derivtwo}{\tau}
\newcommand{\derivthree}{\sigma}

\newcommand{\sctximp}{\hat\sctx}

\newcommand{\tmsix}{q}
\newcommand{\tmp}{\tm'}
\newcommand{\tmtwop}{\tmtwo'}

\newcommand{\prefix}{\prec_p}
\newcommand{\outin}{\prec_O}
\newcommand{\leftright}{\prec_L}
\newcommand{\leftout}{\prec_{LO}}

\newcommand{\lo}{LO}
\newcommand{\lou}{LOU}

\newcommand{\redex}{R}
\newcommand{\redextwo}{Q}
\newcommand{\sizedb}[1]{|#1|_{\db}}
\newcommand{\sizels}[1]{|#1|_{\ls}}

\newcommand{\tostrat}{\leadsto}
\newcommand{\toes}{\leadsto_{X}}
\newcommand{\reflemma}[1]{Lemma~\ref{l:#1}}

\newcommand{\refprop}[1]{Proposition~\ref{prop:#1}}
\newcommand{\refsect}[1]{Sect.~\ref{sect:#1}}

\newcommand{\refpoint}[1]{Point~\ref{p:#1}}
\newcommand{\refremark}[1]{Remark~\ref{rem:#1}}
\newcommand{\refcoro}[1]{Corollary~\ref{coro:#1}}

\newcommand{\ls}{\lssym}

\newcommand{\bctx}{B}
\newcommand{\bctxp}[1]{\bctx\ctxholep{#1}}
\newcommand{\bctxtwo}{B'}

\newcommand{\ap}[2]{#1[#2]}
\renewcommand{\esub}[2]{[#1{\shortleftarrow}#2]}
\renewcommand{\isub}[2]{\{#1{\shortleftarrow}#2\}}

\newcommand{\gctx}{C}
\newcommand{\gctxtwo}{D}
\newcommand{\gctxthree}{E}

\newcommand{\gctxp}[1]{\gctx\ctxholep{#1}}

\newcommand{\gctxthreep}[1]{\gctxthree\ctxholep{#1}}

\renewcommand{\ctx}{S}
\renewcommand{\ctxtwo}{P}
\renewcommand{\ctxthree}{T}
\renewcommand{\ctxfour}{V}

\newcommand{\lsc}{LSC}
\newcommand{\toblo}{\Rew{LO\beta}}
\newcommand{\tolo}{\Rew{LO}}
\newcommand{\tolou}{\Rew{LOU}}
\newcommand{\tollo}{\Rew{LO}}

\renewcommand{\ap}[2]{#1#2}

\usepackage{a4wide}

\begin{document}

\setlength{\pdfpageheight}{\paperheight}
\setlength{\pdfpagewidth}{\paperwidth}

\title{Beta Reduction is Invariant, Indeed\\ (Long Version)}

%%% MACROS DI UGO %%%
\newcommand{\ugo}[1]{\textcolor{blue}{#1}}
\newcommand{\midd}{\; \; \mbox{\Large{$\mid$}}\;\;}

\author{Beniamino Accattoli \and Ugo Dal Lago}

\maketitle

\begin{abstract}
% !TEX root = main.tex 
Slot and van Emde Boas' weak invariance thesis states that
\emph{reasonable} machines can simulate each other within a
polynomially overhead in time. Is $\l$-calculus a reasonable machine?
Is there a way to measure the computational complexity of a $\l$-term?
This paper presents the first complete positive answer to this
long-standing problem. Moreover, our answer is completely
machine-independent and based over a standard notion in the theory of
$\l$-calculus: the length of a leftmost-outermost derivation to normal
form is an invariant cost model. Such a theorem cannot be proved by
directly relating $\l$-calculus with Turing machines or random access
machines, because of the \emph{size explosion problem}: there are
terms that in a linear number of steps produce an exponentially long
output. The first step towards the solution is to shift to a notion of
evaluation for which the length and the size of the output are
linearly related. This is done by adopting the linear substitution
calculus (LSC), a calculus of explicit substitutions modelled after
\linlogic\ \pns\ and admitting a decomposition of leftmost-outermost
derivations with the desired property. Thus, the LSC is invariant with
respect to, say, random access machines. The second step is to show
that LSC is invariant with respect to the $\l$-calculus. The size
explosion problem seems to imply that this is not possible: having the
same notions of normal form, evaluation in the LSC is exponentially
longer than in the $\l$-calculus. We solve such an \emph{impasse} by
introducing a new form of shared normal form and shared reduction,
deemed \emph{useful}. Useful evaluation avoids those steps that
only unshare the output without contributing to
$\beta$-redexes, \ie\ the steps that cause the blow-up in size. The
main technical contribution of the paper is indeed the definition of
useful reductions and the thorough analysis of their properties.

\end{abstract}

% !TEX root = main.tex
%%%%%%%%%%%%%%%%%%%%%%%%
\section{Introduction}
%%%%%%%%%%%%%%%%%%%%%%%%

Theoretical computer science is built around algorithms, computational
models, and machines: an algorithm describes a solution to a problem
with respect to a fixed computational model, whose role is to provide
a handy abstraction of concrete machines. The choice of the model
reflects a tension between different needs. For complexity analysis,
one expects a neat relationship between the model's primitives and the
way in which they are effectively implemented.  In this respect,
random access machines are often taken as the reference model, since
their definition closely reflects the von Neumann architecture. The
specification of algorithms unfortunately lies at the other end of the
spectrum, as one would like them to be as machine-independent as
possible. In this case the typical model is provided by programming
languages. Functional programming languages, thanks to their
higher-order nature, provide very concise and abstract
specifications. Their strength is of course also their weakness: the
abstraction from physical machines is pushed to a level where it is no
longer clear how to measure the complexity of an algorithm. Is there a
way in which such a tension can be solved?

The tools for stating the question formally are provided by complexity
theory and by Slot and van Emde Boas' invariance thesis
\cite{DBLP:conf/stoc/SlotB84}:
\begin{center}
  \emph{Reasonable} computational models simulate each other\\ with
  polynomially bounded overhead in time,\\ and constant factor overhead
  in space.
\end{center}
The \emph{weak} invariance thesis is the variant where the requirement
about space is dropped, and it is the one we will actually work with
in this paper. The idea behind the thesis is that for reasonable
models the definition of every polynomial or super-polynomial class
such as $\mathbf{P}$ or $\mathbf{EXP}$ does not rely on the chosen
model. On the other hand, it is well-known that sub-polynomial classes
depend very much on the model, and thus it does not make much sense to
pursue a linear rather than polynomial relationship.
  
A first refinement of our question then is: are functional languages
invariant with respect to standard models like random access machines
or Turing machines? Such an invariance has to be proved via an
appropriate measure of time complexity for programs, \ie\ a \emph{cost
  model}.

The natural answer is to consider the \emph{unitary} cost model, \ie\
take the number of evaluation steps as the cost of the underlying
term. However, this is not well-defined.  The evaluation of functional
programs, indeed, depends very much on the evaluation strategy chosen
to implement the language, as the $\l$-calculus --- the reference
model for functional languages --- is so machine-independent that it
does not even come with a deterministic evaluation strategy.  And
which strategy, if any, gives us the most natural, or \emph{canonical}
cost model (whatever that means)?  These questions have received some
attention in the last decades.  The number of optimal parallel
$\beta$-steps (in the sense of \levy\ \cite{thesislevy}) to normal
form has been shown \emph{not} to be a reasonable cost model: there
exists a family of terms which reduces in a polynomial number of
parallel $\beta$-steps, but whose complexity is
non-elementary~\cite{DBLP:conf/icfp/LawallM96,DBLP:conf/popl/AspertiM98}. If
one considers the number of \emph{sequential} $\beta$-steps (in a
given strategy, for a given notion of reduction), the literature
offers some partial positive results, all relying on the use of
sharing (see below for more details).  Some quite general
results~\cite{DBLP:journals/corr/abs-1208-0515,DBLP:conf/rta/AvanziniM10}
have been obtained through graph rewriting, itself a form of sharing,
when only first order symbols are considered.

Sharing is indeed a key ingredient, for one of the issues here is due
to the \emph{representation of terms}. The ordinary way of
representing terms indeed suffers from the \emph{size explosion
  problem}: even for the most restrictive notions of reduction
(e.g. Plotkin's weak reduction), there is a family of terms
$\{\tm_n\}_{n\in\nat}$ such that $|\tm_n|$ is linear in $n$, $\tm_n$
evaluates to its normal form in $n$ steps, but at $i$-th step a term
of size $2^i$ is copied, so that the size of the normal form of
$\tm_n$ is exponential in $n$. Put differently, an evaluation sequence
of linear length can possibly produce an output of exponential
size. At first sight, then, there is no hope that evaluation lengths
may provide an invariant cost model. The idea is that such an
\emph{impasse} can be avoided by sharing common sub-terms along the
evaluation process, in order to keep the representation of the output
compact. But is appropriately managed sharing enough? The literature
offers some positive, but partial, answers to this question. The
number of steps is indeed known to be an invariant cost model for weak
reduction~\cite{DBLP:journals/tcs/LagoM08,DBLP:journals/corr/abs-1208-0515}
and for head reduction~\cite{DBLP:conf/rta/AccattoliL12}.

If the problem at hand consists in computing the \emph{normal form} of
an arbitrary $\lambda$-term, however, no positive answer is known. We
believe that not knowing whether the $\lambda$-calculus in its full
generality is a reasonable machine is embarrassing for the
$\lambda$-calculus community. In addition, this problem is relevant
in practice: proof assistants often need to check whether two terms
are convertible, itself a problem that can be reduced to the one under
consideration.

In this paper, we give a positive answer to the question above, by
showing that leftmost-outermost (\lo, for short) reduction \emph{to
  normal form} indeed induces an invariant cost model. Such an
evaluation strategy is \emph{standard}, in the sense of the
standardisation theorem, one of the central theorems in the theory of
$\l$-calculus, first proved by Curry and Feys
\cite{curry1958combinatory}. The relevance of our cost model is given
by the fact that \lo\ reduction is an abstract concept from rewriting
theory which at first sight is totally unrelated to complexity
analysis. In particular, our cost model is completely
machine-independent.

Another view on this problem comes in fact from rewriting theory
itself. It is common practice to specify the operational semantics of
a language via a rewriting system, whose rules always employ some form
of substitution, or at least of copying, of subterms. Unfortunately,
this practice is very far away from the way languages are
implemented. Indeed, actual interpreters perform copying in a very
controlled way (see, \eg,
\cite{Wad:SemPra:71,Pey:ImplFunProgLang:87}). This discrepancy
induces serious doubts about the relevance of the computational
model. Is there any theoretical justification for copy-based models,
or more generally for rewriting theory as a modelling tool? In this
paper we give a very precise answer, formulated within rewriting
theory itself.

As in our previous work \cite{DBLP:conf/rta/AccattoliL12}, we prove
our result by means of the \emph{linear substitution calculus} (see
also \cite{DBLP:conf/rta/Accattoli12,non-standard-preprint}), a simple
calculus of explicit substitutions (ES, for short) arising from linear
logic and graphical syntaxes and similar to calculi studied by De
Bruijn~\cite{deBruijn87}, Nederpelt~\cite{Ned92}, and
Milner~\cite{DBLP:journals/entcs/Milner07}. A peculiar feature of the
linear substitution calculus (\lsc) is the use of rewriting rules
\emph{at a distance}, \ie\ rules defined by means of contexts, that
are used to closely mimic reduction in linear logic proof nets. Such a
framework --- whose use does not require any knowledge of these areas
--- allows an easy management of sharing and, in contrast to previous
approaches to ES, admits a theory of standardisation and a notion of
\lo\ evaluation~\cite{non-standard-preprint}. The proof of our result
indeed is a \emph{tour de force} based on a fine quantitative study of
the relationship between \lo derivations for the $\l$-calculus and a
variation over \lo\ derivations for the \lsc. Roughly, the latter
avoids the size explosion problem while keeping a polynomial
relationship with the former.

Let us point out that invariance results usually have two directions,
while we here study only one of them (namely that the $\l$-calculus
can be efficiently simulated by, say, Turing machines). The missing
half is a much simpler problem already solved in
\cite{DBLP:conf/rta/AccattoliL12}: there is an encoding of Turing
machines into $\l$-terms s.t. their execution is simulated by weak
head $\beta$-reduction with only a linear overhead.

\emph{On Invariance and Complexity Analysis.}  
Before proceeding, let us stress some crucial points:
\begin{varenumerate}
\item 
  \emph{ES Are Only a Tool}. Although ES are an essential
  \emph{tool} for the proof of our result, the \emph{result itself} is
  about the usual, pure, $\lambda$-calculus. In particular, the
  invariance result can be used without any need to care about ES: we
  are allowed to measure the complexity of problems by simply bounding
  the \emph{number} of \lo\ $\beta$-steps taken by any $\lambda$-term
  solving the problem.
\item 
  \emph{Complexity Classes in the $\l$-Calculus}. The main
  consequence of our invariance result is that every polynomial or
  super-polynomial class, like $\mathbf{P}$ of $\mathbf{EXP}$, can be
  defined using $\l$-calculus (and \lo\ $\beta$-reduction) instead of
  Turing machines.
\item 
  \emph{Our Cost Model is Unitary}. An important point is that our
  cost model is \emph{unitary}, and thus attributes a constant cost to
  any \lo\ step. One could argue that it is always possible to reduce
  $\lambda$-terms on abstract or concrete machines and take that
  number of steps as \emph{the} cost model. First, such a measure of
  complexity would be very machine-dependent, against the very essence
  of $\l$-calculus. Second, these cost models invariably attribute a
  more-than-constant cost to any $\beta$-step, making the measure much
  harder to use and analyse. It is not evident that a computational
  model enjoys a unitary invariant cost model. As an example, if
  multiplication is a primitive operation, random access machines need
  to be endowed with a \emph{logarithmic} cost model in order to
  obtain invariance.
\end{varenumerate}
The next section explains why the problem at hand is hard, and in
particular why iterating our previous results on head reduction
\cite{DBLP:conf/rta/AccattoliL12} does not provide a solution.

% !TEX root = main.tex
%%%%%%%%%%%%%%%%%%%%%%%%%%%%%%%%%%%%%%%
\section{Why is The Problem Hard?}
%%%%%%%%%%%%%%%%%%%%%%%%%%%%%%%%%%%%%%%
In principle, one may wonder why sharing is needed at all, or whether
a relatively simple form of sharing suffices. In this section, we will
show that sharing is unavoidable and that a new subtle notion of
sharing is necessary.

If we stick to explicit representations of terms, in which sharing is
not allowed, counterexamples to invariance can be designed in a fairly
easy way.  Let $\tmtwo$ be the lambda term $\vartwo\var\var$ and
consider the sequence $\{\tm_n\}_{n\in\nat}$ of $\lambda$-terms
defined as $\tm_0=\tmtwo$ and $\tm_{n+1}=(\lambda\var.\tm_n)\tmtwo$
for every $n\in\nat$.  The term $\tm_n$ has size linear in $n$, and
$\tm_n$ rewrites to its normal form $\tmthree_n$ in exactly $n$ steps,
following the \lo\ reduction order; as an example:
\begin{align*}
\tm_0&=\tmtwo=\tmthree_0;\\
\tm_1&\rightarrow\vartwo\tmtwo\tmtwo=\vartwo\tmthree_0\tmthree_0=\tmthree_1;\\
\tm_2&\rightarrow(\lambda\var.\tm_0)(\vartwo\tmtwo\tmtwo)=(\lambda\var.\tmtwo)\tmthree_1
  \rightarrow\vartwo\tmthree_1\tmthree_1=\tmthree_2.
%&\vdots
\end{align*}
For every $n$, however, $\tmthree_{n+1}$ contains two copies of
$\tmthree_n$, hence the size of $\tmthree_n$ is \emph{exponential} in
$n$. As a consequence, the unitary cost model \emph{is not} invariant:
in a linear number of $\beta$-steps we reach an object which cannot
even be written down in polynomial time.

The solution the authors proposed in~\cite{DBLP:conf/rta/AccattoliL12}
is based on ES, and allows to tame the size explosion problem in a
satisfactory way when \emph{head} reduction suffices. In particular,
the head steps above become the following \emph{linear} head steps:
\begin{align*}
\tm_0&=\tmtwo=\tmfour_0;\\
\tm_1&\rightarrow(\vartwo\var\var)\esub\var\tmtwo=\tmtwo\esub\var\tmtwo=\tmfour_1;\\
\tm_2&\rightarrow((\lambda\var.\tm_0)\tmtwo)\esub\var\tmtwo=((\lambda\var.\tmtwo)\tmtwo)\esub\var\tmtwo\\
  &\rightarrow\tmtwo\esub\var\tmtwo\esub\var\tmtwo=\tmfour_2.
%&\vdots
\end{align*}
As one can easily verify, the size of $\tmfour_n$ is linear in
$n$. More generally, linear head reduction (LHR) has the \emph{subterm
  property}, \ie\ it only duplicates subterms of the initial
term. This fact implies that the size of the result and the length of
the derivation are linearly related. In other words, the size
explosion problem has been solved. Of course one needs to show that 1)
the compact results \emph{unfold} to the expected result (that may be
exponentially bigger), and 2) that compact representations can be
managed efficiently (typically they can be tested for equality in time
polynomial in the size of the compact representation),
see~\cite{DBLP:conf/rta/AccattoliL12} or below for more details.

It may seem that one is then forced to use ES to measure
complexity. In~\cite{DBLP:conf/rta/AccattoliL12} we also showed that
LHR is at most quadratically longer than head reduction, so that the
polynomial invariance of LHR lifts to head reduction. This is how we
exploit sharing to circumvent the size explosion problem: we are
allowed to take the length of the head derivation as a cost model,
even if it suffers of the size explosion problem, because the actual
implementation is meant to be done via LHR and be only polynomially
(actually quadratically) longer.

There is a natural candidate for extending the approach to reduction
to \emph{normal form}: just iterate the (linear) head strategy on the
arguments, obtaining the (linear) \lo\ strategy, that \emph{does}
compute normal forms \cite{non-standard-preprint}. As we will show,
for linear \lo\ derivations the subterm property holds. The size of
the output is still under control, being linearly related to the
length of the \lo\ derivation. Unfortunately, when computing normal
forms this is not enough.

One of the key points in our previous work was that there is a notion
of \emph{linear head} normal form that is a compact representation for
\emph{head} normal forms. The generalisation of such an approach to
normal forms has to face a fundamental problem: what is a
\emph{linear} normal form? Indeed, terms with and without ES share the
same notion of normal form. Consider again the family of terms
$\{\tm_n\}_{n\in\nat}$: if we go on and unfold all substitutions in
$\tmfour_n$, we end up in $\tmthree_n$. Thus, by the subterm property,
the linear \lo\ strategy takes an exponential number of steps, and so
it cannot be polynomially related to the \lo\ strategy.

Summing up, we need a strategy that 1) implements the \lo\ strategy,
2) has the subterm property and 3) never performs \emph{useless}
substitution steps, \ie\ those steps whose role is simply to explicit
the normal form, without contributing in any way to
$\beta$-redexes. The main contribution of this work is the definition
of such a linear \emph{useful} strategy, and the proof that it is
indeed polynomially related to both the \lo\ strategy and a concrete
implementation model.

This is not a trivial task, actually. One may think that it is enough
to evaluate a term $\tm$ in a \lo\ way, stopping as soon as the
unfolding $\unf{\tmtwo}$ of the current term $\tmtwo$ --- the
term obtained by expanding the ES of $\tmtwo$ --- is a $\beta$-normal
form. Unfortunately, this simple approach does not work, because the
exponential blow-up may be caused by ES \emph{lying between} two
$\beta$-redexes, so that proceeding in a \lo\ way would
\cben{execute}{unfold} the problematic substitutions anyway.

Our notion of \emph{useful} step \ben{will} elaborate on this idea, by
computing \emph{partial} unfoldings, to check if a substitution step
contributes or will contribute to some future $\beta$-redex. Of
course, we will have to show that such tests can be themselves
perfomed in polynomial time, and that the notion of \lo\ \emph{useful}
reduction retains all the good properties of \lo\ reduction.

% !TEX root = main.tex
%%%%%%%%%%%%%%%%%%%%%%%%
\section{The Calculus}
%%%%%%%%%%%%%%%%%%%%%%%%
We assume familiarity with the $\l$-calculus (see
\cite{Barendregt84}).  The language of the \emph{linear substitution
  calculus} (\lsc\ for short) is given by the following grammar for
terms:
$$
\tm,\tmtwo,\tmthree,\tmfour\grameq\var\midd \l \var. \tm \midd \ap\tm \tmtwo\midd  \tm\esub\var\tmtwo.
$$ 
The constructor $\tm\esub{\var}{\tmtwo}$ is called an \emph{explicit
  substitution} (of $\tmtwo$ for $\var$ in $\tm$, the usual (implicit)
substitution is instead noted $\tm\isub{\var}{\tmtwo}$). Both $\l
\var. \tm$ and $\tm\esub{\var}{\tmtwo}$ bind $\var$ in $\tm$, and we
silently work modulo $\alpha$-equivalence of these bound variables,
\eg\ $(\var\vartwo)\esub\vartwo\tm\isub\var\vartwo =
(\vartwo\varthree)\esub\varthree\tm$. We use $\fv\tm$ for the set of
free variables of $\tm$.

\emph{Contexts.} The operational semantics of the \lsc\ is parametric
in a notion of (one-hole) context. General contexts are defined by:
$$
\gctx \grameq \ctxhole\midd \l \var. \gctx\midd \ap\gctx \tm \midd\ap\tm\gctx\midd\gctx\esub{\var}{\tm}\midd\tm\esub{\var}{\gctx},
$$ 
and the plugging of a term $\tm$ into a context $\gctx$ is defined as
$\ctxhole\ctxholep\tm\defeq\tm$, $(\l\var.\gctx)\ctxholep\tm \defeq
\l\var.(\gctx\ctxholep\tm)$, and so on. As usual, plugging in a
context can capture variables,
\eg\ $((\ctxhole\vartwo)\esub\vartwo\tm)\ctxholep\vartwo =
(\vartwo\vartwo)\esub\vartwo\tm$. The plugging $\gctxp\gctxtwo$ of a
context $\gctxtwo$ into a context $\gctx$ is defined analogously.

Along most of the paper, however, we will not need such a general
notion of context. In fact, our study takes a simpler form if the
operational semantics is defined with respect to \deff{shallow}
contexts, defined as (note the absence of the production
$\tm\esub{\var}{\ctx}$):
$$
\ctx,\ctxtwo,\ctxthree,\ctxfour \grameq \ctxhole\midd \l \var. \ctx\midd \ap\ctx \tm \midd\ap\tm\ctx\midd\ctx\esub{\var}{\tm}.
$$
In the following, whenever we refer to a \emph{context} without
further specification it is implicitly assumed that it is a
\emph{shallow} context. A special class of contexts is that of
\deff{substitution contexts}: 
$$
\sctx\grameq\ctxhole\midd
\sctx\esub{\var}{\tm}.
$$

\emph{Operational Semantics.} The (shallow) rewriting rules $\todb$
($\db$ = \emph{$\beta$ at a distance}) and $\tols$ (linear
substitution) are given by the closure by (shallow) contexts of the
following rules:
\begin{align*}
  \ap{\sctxp{\l \var.\tm}} \tmtwo & \rtodb \sctxp{\tm\esub{\var}{\tmtwo}};\\
  \ctxp{\var}\esub{\var}{\tmtwo} &\rtols \ctxp{\tmtwo}\esub{\var}{\tmtwo}.
\end{align*}
The union of $\todb$ and $\tols$ is simply noted $\to$. The rewriting
rules are assumed to use \emph{on-the-fly} $\alpha$-equivalence to
avoid variable capture. For instance,
\begin{align*}
    (\l \var.\tm)\esub\vartwo\tmtwo \vartwo & \todb
  \tm\isub\vartwo\varthree\esub\var\vartwo\esub\varthree\tmtwo\qquad\mbox{for
    $\varthree\notin\fv\tm$};\\ (\l
  \vartwo.(\var\vartwo))\esub{\var}{\vartwo} &\tols (\l
  \varthree.(\vartwo\varthree))\esub{\var}{\vartwo}.
\end{align*}
Moreover, in rule $\ls$ the context $\ctx$ is assumed to not capture
$\var$, so that $(\l \var.\var)\esub{\var}{\vartwo} \not\tols (\l
\var.\vartwo)\esub{\var}{\vartwo}$.

The just defined shallow fragment simply ignores garbage collection
(that in the \lsc\ can always be postponed
\cite{DBLP:conf/rta/Accattoli12}) and lacks some of the nice
properties of the LSC (obtained simply by replacing shallow contexts
by general contexts). Its relevance is the fact that it is the
smallest fragment implementing linear \lo\ reduction. The following
are examples of shallow steps:
\begin{align*}
  \ap{(\l\var.\var)}\vartwo&\todb\var\esub\var\vartwo;\\
  (\ap\var\var)\esub\var\tm&\tols(\ap\var\tm)\esub\var\tm;
\end{align*}
while the following steps are not 
\begin{align*}
  \tm\esub\varthree{\ap{(\l\var.\var)}\vartwo}&\todb\tm\esub\varthree{\var\esub\var\vartwo};\\
  \var\esub\var\vartwo\esub\vartwo\tm&\tols\var\esub\var\tm\esub\vartwo\tm.
\end{align*}

Taking the external context into account, a substitution step has the
following \emph{explicit} form: $
\ctxtwop{\ctxp{\var}\esub{\var}{\tmtwo}} \tols
\ctxtwop{\ctxp{\tmtwo}\esub{\var}{\tmtwo}}$. We shall often use a
\emph{compact} form, writing $\ctxthreep{\var}\tols\ctxthreep{\tmtwo}$
where it is implicitly assumed that $\ctxthree =
\ctxtwop{\ctx\esub{\var}{\tmtwo}}$. We use $\redex$ and $\redextwo$ as
metavariables for redexes. A \deff{derivation} $\deriv:\tm\to^k\tmtwo$
is a finite sequence of reduction steps, sometimes given as
$\redex_1;\ldots; \redex_k$, \ie\ as the sequence of reduced
redexes. We write $\size\tm$ for the size of $\tm$, $\esmeas\tm$ for
the number of substitutions in $\tm$, $\size\deriv$ for the length of
$\deriv$, and $\sizedb\deriv$ for the number of $\db$-steps in
$\deriv$.

\emph{(Relative) Unfoldings.} The unfolding $\unf{\tm}$ of a term
$\tm$ is the $\l$-term obtained from $\tm$ by turning its explicit
substitutions into implicit ones:
\begin{align*}
  \unf{\var}&\defeq\var; &
  \unf{(\ap\tm \tmtwo)}&  \defeq \ap{\unf{\tm}} {\unf{\tmtwo}};\\
  \unf{(\l \var. \tm)} & \defeq  \l \var. \unf{\tm};&
  \unf{(\tm\esub{\var}{\tmtwo})} & \defeq  \unf{\tm}\isub{\var}{\unf{\tmtwo}}.
\end{align*}
We will also need a more general notion, the unfolding
$\relunf{\tm}{\ctx}$ of $\tm$ in a context $\ctx$:
\begin{align*}
\relunf{\tm}{\ctxhole}&\defeq\unf{t};&\relunf{\tm}{\tmtwo\ctx}&\defeq\relunf{\tm}{\ctx};
    &\relunf{\tm}{\ctx\esub{\var}{\tmtwo}}&\defeq\relunf{\tm}{\ctx}\isub{\var}{\unf{\tmtwo}};\\
\relunf{\tm}{\l \var. \ctx} &\defeq\relunf{\tm}{\ctx};&\relunf{\tm}{\ctx\tmtwo}&\defeq\relunf{\tm}{\ctx}.
\end{align*}
For instance, 
\begin{align*}
  \unf{(\ap\var{(\ap\vartwo\varthree))}\esub\vartwo\var\esub\var\varthree} 
  &=
  \ap\varthree{(\ap\varthree\varthree)};\\		
  \relunf{(\ap\var\vartwo)}{(\ap{\ctxhole\esub\vartwo\var}\tm)\esub\var{\l\varthree.(\ap\varthree\varthree)}}
  &=
  {(\l\varthree.(\ap\varthree\varthree))}{\l\varthree.(\ap\varthree\varthree)}.
\end{align*}
We extend implicit substitutions and unfoldings to contexts by setting
$\ctxhole\isub\var\tm\defeq\ctxhole$ and
$\unf{\ctxhole}\defeq\ctxhole$ (all other cases are defined as
expected, \eg\ $\unf{\ctx\esub\var\tm}\defeq
\unf{\ctx}\isub\var{\unf{\tm}}$). We also write $\ctx\prefix\tm$ if
there is a term $\tmtwo$ s.t. $\ctxp\tmtwo=\tm$, call it the
\deff{prefix relation}. We have the following properties, that only hold
because our contexts are shallow (implying that the hole cannot
be duplicated during the unfolding).

\begin{lemma}
\label{l:ctx-unf} %\reflemma{ctx-unf}.\ref{p:ctx-unf-three}
Let $\ctx$ be a shallow context. Then:
\begin{varenumerate}
\item\label{p:ctx-unf-one}
  $\unf{\ctx}$ is a shallow context;
\item\label{p:ctx-unf-two}
  $\ctxp\tm\isub\var\tmtwo={\ctx\isub\var\tmtwo}\ctxholep{\tm\isub\var\tmtwo}$;
\item\label{p:ctx-unf-three}
  $\unf{\ctxp\tm}=\unf{\ctx}\ctxholep{\relunf{\tm}{\ctx}}$, in
  particular if $\ctx\prefix\tm$ then $\unf{\ctx}\prefix\unf{\tm}$.
\end{varenumerate}
\end{lemma}

\begin{proof}
  All points are by induction on $\ctx$. In particular,
  \refpoint{ctx-unf-three} uses \refpoint{ctx-unf-two} when
  $\ctx=\ctxtwo\esub\var\tmtwo$.
\end{proof}

Given a derivation $\deriv:\tm\to^*\tmtwo$ in the \lsc, we often
consider the $\beta$-derivation
$\unf{\deriv}:\unf{\tm}\tob^*\unf{\tmtwo}$ obtained by projecting
$\deriv$ via unfolding.

\newcommand{\str}{\rightarrow} 
\newcommand{\nos}[2]{\#_{#1}(#2)}
\newcommand{\strp}[3]{#1^{#2}\!\!(#3)}
\emph{Reduction Combinatorics.} Given any calculus, a deterministic
strategy $\str$ for it, and a term $\tm$, the expression
$\nos{\str}{\tm}$ stands for the number of reduction steps necessary
to reach the normal form of $\tm$ along $\str$, or $\infty$ if $\tm$
diverges. Similarly, given a natural number $n$, the expression
$\strp{\str}{n}{\tm}$ stands for the term $\tmtwo$ such that
$\tm\str^n\tmtwo$, if $n\leq\nos{\str}{\tm}$, or for the normal form
of $\tm$ otherwise.

% !TEX root = main.tex
%%%%%%%%%%%%%%%%%%%%%%%%%%%%%%%%%%%%%%
\section{The Proof, Made Abstract}
\label{sect:abstract-proof}
%%%%%%%%%%%%%%%%%%%%%%%%%%%%%%%%%%%%%%
Our proof method can be described abstractly. Such an approach both
clarifies the structure of the proof and prepares the ground for
possible generalisations to, \eg, the call-by-value $\l$-calculus or
calculi with additional features as pattern matching or control
operators. We want to show that a certain strategy $\tostrat$ for the
$\l$-calculus provides a unitary and invariant cost model, \ie\ that
the number of $\tostrat$ steps is a measure polynomially related to
the number of transitions on Turing machines. As explained in the
introduction, we pass through an intermediary computational model, a
calculus with ES, the \emph{linear substitution calculus}, playing the
role of a very abstract machine for $\l$-terms.

We are looking for an appropriate strategy $\toes$ within the LSC
which is invariant with respect to both $\tostrat$ and Turing
machines. Then we need two theorems, which together form the main
result of the paper:
\begin{varenumerate}
\item 
  \emph{High-Level Implementation}: $\tostrat$ terminates iff
  $\toes$ terminates. Moreover, $\tostrat$ is implemented by $\toes$
  with only a polynomial overhead. Namely, $\tm\toes^k\tmtwo$ iff
  $\tm\tostrat^h\unf{\tmtwo}$ with $k$ polynomial in $h$ (our actual
  bound will be quadratic);
\item 
  \emph{Low-Level Implementation}: $\toes$ is implemented on
  Turing machines with an overhead in time which is polynomial in both
  $k$ and the size of $\tm$.
\end{varenumerate}
The high-level part relies on the following notion.
\begin{definition}
Let $\tostrat$ be a deterministic strategy on $\l$-terms and $\toes$ a
strategy of the \lsc. The pair $(\tostrat,\toes)$ is a
\deff{high-level implementation system} if whenever $\tm$ is a
$\l$-term and $\deriv:\tm\toes^*\tmtwo$ then:
\begin{varenumerate}
  \item 
    \emph{Normal Form}: if $\tmtwo$ is a $\toes$-normal form then
    $\unf{\tmtwo}$ is a $\tostrat$-normal form.
  \item 
    \emph{Projection}: $\unf{\deriv}:\tm\tostrat^*\unf{\tmtwo}$ and
    $\size{\unf{\deriv}}=\sizedb\deriv$.
  \item 
    \emph{Trace}: the number $\esmeas\tmtwo$ of ES in $\tmtwo$ is
    exactly the number $\sizedb\deriv$ of $\db$-steps in $\deriv$;
  \item 
    \emph{Syntactic Bound}: the length of a sequence of substitution
    steps from $\tmtwo$ is bounded by $\esmeas\tmtwo$.
\end{varenumerate}
\end{definition}
Concretely, the high-level implementation system at work in the paper
will take as $\tostrat$ the \lo\ strategy of the $\l$-calculus and as
$\toes$ \emph{a variant} of the linear \lo\ strategy for the \lsc. A
variant is required because, as we will explain, the linear
\lo\ strategy of the \lsc\ does not satisfy the syntactic bound
property.

The normal form and projection properties address the
\emph{qualitative} part of the high-level implementation theorem, \ie\
the part about termination. The normal form property guarantees that
$\toes$ does not stop prematurely, so that when $\toes$ terminates
$\tostrat$ cannot keep going. The projection property guarantees that
termination of $\tostrat$ implies termination of $\toes$. The two
properties actually state a stronger fact: \emph{$\tostrat$ steps can
  be identified with the $\db$-steps of the $\toes$ strategy}.

The trace and syntactic bound properties are instead used for the
\emph{quantitative} part of the theorem, \ie\ to provide the
polynomial bound. The two properties together provide a bound the
number of $\ls$-steps in a $\toes$ derivation with respect to the
number of $\db$-steps, that---by the identification of $\beta$ and
$\db$ redexes---is exactly the length of the associated $\tostrat$
derivation.

The high-level part can now be proved abstractly.
\begin{theorem}[High-Level Implementation]
  Let $\tm$ be an ordinary $\l$-term and $(\tostrat,\toes)$ a
  high-level implementation system. Then:
  \begin{varenumerate}
  \item 
    $\tm$ is $\tostrat$-normalising iff it is $\toes$-normalising.
  \item 
    If $\deriv:\tm\toes^*\tmtwo$ then
    $\unf{\deriv}:\tm\tostrat^*\unf{\tmtwo}$ and
    $\size\deriv=O(\size{\unf{\deriv}}^2)$.
  \end{varenumerate}
\end{theorem}
\proof
\begin{varenumerate}
\item 
  $\Leftarrow$) Suppose that $\tm$ is $\toes$-normalisable and let
  $\deriv:\tm\toes^*\tmtwo$ a derivation to $\toes$-normal form. By
  the projection property there is a derivation
  $\tm\tostrat^*\unf{\tmtwo}$. By the normal form property
  $\unf{\tmtwo}$ is a $\tostrat$-normal form.
  
  $\Rightarrow$) Suppose that $\tm$ is $\tostrat$-normalisable and let
  $\derivtwo:\tm\tostrat^k\tmtwo$ be the derivation to
  $\tostrat$-normal form (unique by determinism of
  $\tostrat$). Assume, by contradiction, that $\tm$ is not
  $\toes$-normalisable. Then there is a family of $\toes$-derivations
  $\deriv_i:\tm\toes^i\tmtwo_i$ with $i\in\nat$, each one extending
  the previous one. By the syntactic bound property, $\toes$ can make
  only a finite number of $\ls$ steps (more generally, $\tols$ is
  strongly normalising in the \lsc). Then the sequence
  $\set{\sizedb{\deriv_i}}_{i\in\nat}$ is non-decreasing and
  unbounded. By the projection property, the family
  $\set{\deriv_i}_{i\in\nat}$ unfolds to a family of
  $\tostrat$-derivations $\set{\unf{\deriv_i}}_{i\in\nat}$ of
  unbounded length (in particular greater than $k$), absurd.
  
\item 
  By the projection property, it follows that
  $\unf{\deriv}:\tm\tostrat^*\unf{\tmtwo}$. Moreover, to show
  $\size\deriv=O(\size{\unf{\deriv}}^2)$ it is enough to show
  $\size\deriv=O(\sizedb{\deriv}^2)$. Now, $\deriv$ has the shape:
  \begin{center}
    {\footnotesize
      $\tm=\tmthree_1\todb^{a_1}\tmfour_1\tols^{b_1}\tmthree_2\todb^{a_2}\tmfour_2\tols^{b_2}\ldots\tmthree_k\todb^{a_k}\tmfour_k\tols^{b_k}\tmtwo$.}
  \end{center}
  By the syntactic bound property, we obtain
  $b_i\leq\esmeas{\tmfour_i}$.  By the trace property we obtain
  $\esmeas{\tmfour_i}=\sum_{j=1}^i a_j$, and so $b_i\leq\sum_{j=1}^i
  a_j$. Then:
  \begin{center}
    $\sizels\deriv=\sum_{i=1}^k b_i\leq\sum_{i=1}^k\sum_{j=1}^i a_j$.
  \end{center}
  Note that $\sum_{j=1}^i a_j\leq \sum_{j=1}^k a_j=\sizedb\deriv$ and
  $k\leq\sizedb\deriv$. So
  \begin{center}
    $\sizels\deriv\leq\sum_{i=1}^k\sum_{j=1}^i a_j\leq\sum_{i=1}^k\sizedb\deriv\leq \sizedb\deriv^2$.
  \end{center}
  Finally, $\size\deriv=\sizedb\deriv+\sizels\deriv\leq
  \sizedb\deriv+\sizedb\deriv^2=O(\sizedb\deriv^2)$.\qed
\end{varenumerate}
\medskip
\noindent For the low-level part we rely on the following notion.

\begin{definition}
\label{def:mech}
A strategy $\toes$ on \lsc\ terms is \deff{mechanisable} if given a
derivation $\deriv:\tm\toes^*\tmtwo$:
\begin{varenumerate}
\item 
  \emph{Subterm}: the terms duplicated along $\deriv$ are subterms of
  $\tm$.
\item 
  \emph{Selection}: the search of the next $\toes$ redex to reduce in
  $\tmtwo$ takes polynomial time in $\size\tmtwo$.
\end{varenumerate}
\end{definition}
The \emph{subterm property} \ugo{---} essentially \ugo{---} guarantees
that any step has a linear cost in the size of the initial term, the
fundamental parameter for complexity; it will be discussed in more
detail in \refsect{standard}. At first sight the \emph{selection
  property} is always trivially verified: finding a redex in $\tmtwo$
takes time linear in $\size\tmtwo$. However, our strategy for ES will
reduce only redexes satisfying a side-condition whose na\"ive
verification is exponential in $\size\tmtwo$. Then one has to be sure
that such a computation can be done in polynomial time.

\begin{theorem}[Low-Level Implementation]
  Let $\toes$ be a mechanisable strategy.  Then there is an algorithm
  that on input $\tm$ and $k$ outputs $\strp{\toes}{k}{\tm}$, and
  which works in time polynomial in $k$ and $\size\tm$.
\end{theorem}
\begin{proof}
  By the subterm property, implementing one step takes time polynomial
  (if not linear) in $\size\tm$. An immediate consequence of the
  subterm property is the \emph{no size explosion property}, \ie\ that
  $\size\tmtwo\leq(k+1)\cdot\size\tm$. By the selection property
  selecting the next redex takes time polynomial in $\size\tmtwo$,
  that by the no size explosion property is polynomial in $k$ and
  $\size\tm$. The composition of polynomials is again a polynomial,
  and so selecting the redex takes time polynomial in $k$ and
  $\size\tm$. Hence, the reduction can be implemented in polynomial
  time.
\end{proof}

In \cite{DBLP:conf/rta/AccattoliL12}, we proved that \emph{head
  reduction} and \emph{linear head reduction} form a high-level
implementation system and that linear head reduction is mechanisable,
even if we did not use such a terminology, nor were we aware of the
presented abstract scheme. In order to extend such a result to normal
forms we need to replace head reduction with a normalising strategy
(\ie\ a strategy reaching the $\beta$-normal form, if any).

One candidate for $\tostrat$ is the \lo\ strategy $\toblo$. Such a
choice is natural, as $\toblo$ is normalising, it produces standard
derivations, and it is an iteration of head reduction. What is left to
do, then, is to find a strategy $\toes$ for ES, which is both
mechanisable and a high-level implementation of
$\toblo$. Unfortunately, the linear \lo\ strategy, here noted $\tollo$
and first defined in \cite{non-standard-preprint}, is mechanisable but
the pair $(\toblo,\tollo)$ is not a high-level implementation system.

In general, mechanisable strategies are not hard to find. As we will
show in \refsect{standard}, the whole class of \emph{standard}
derivations for ES has the subterm property. In particular, the linear
strategy $\tollo$ --- which is standard --- enjoys all the other
properties \emph{but} for the syntactic bound property.

Such a problem will be solved by \lo\ \emph{useful} derivations, to be
introduced in \refsect{useful}, that will be shown to be both
mechanisable and a high-level implementation of $\toblo$. Useful
derivations avoid those substitution steps that only explicit the
normal form without contributing to explicit $\beta/\db$-redexes
(that, by the projection property, can be identified). \lo\ useful
derivations will have all the nice properties of \lo\ derivations and
moreover will stop on shared, minimal representations of normal forms,
solving the problem with linear \lo\ derivations.

Let us point out that our analysis would be vacuous without evidence
that useful normal forms are a reasonable representation of
$\l$-terms. In other words, we must be sure that ES do not hide (too
much of) the inherent difficulty of reducing $\lambda$-terms under the
carpet of sharing. In~\cite{DBLP:conf/rta/AccattoliL12}, we solved
this issue by providing an efficient algorithm for checking the
equality of any two \lsc\ terms --- thus in particular of useful normal
forms --- without computing their unfoldings (that otherwise would
reintroduce an exponential blow-up). Some further discussion can be
found in Sect.~\ref{sect:properties}.

% !TEX root = main.tex
%%%%%%%%%%%%%%%%%%%%%%%%%%%%%%%%%%%%%%%%%%%%%%%%%%%%%%%%%%%%%%
\section{Useful Derivations}
\label{sect:useful}
%%%%%%%%%%%%%%%%%%%%%%%%%%%%%%%%%%%%%%%%%%%%%%%%%%%%%%%%%%%%%%
In this section we define a constrained, optimised notion of
reduction, that will be the key to the High-Level Implementation
Theorem. The idea is that an optimised step takes place only if it
somehow contributes to explicit a $\beta/\db$-redex. Let an
\deff{applicative context} be defined by $\apctx \grameq
\ctxp{\sctx\tm}$, where $\ctx$ and $\sctx$ are a shallow and a
substitution context, respectively (note that applicative contexts are
\emph{not} made out of applications only; for instance $\tm\l
\var.(\ctxhole \esub\vartwo\tmtwo \tmthree)$ is an applicative
context). Then:
\begin{definition}[Useful/Useless Steps and Derivations]
  A \deff{useful} step is either a $\db$-step or a $\ls$-step
  $\ctxp{\var}\tols \ctxp{\tmthree}$ (in compact form)
  s.t. $\relunf{\tmthree}{\ctx}$:
  \begin{varenumerate}
  \item 
    either contains a $\beta$-redex,
  \item 
    or is an abstraction and $\ctx$ is an applicative context.
  \end{varenumerate}
  A \deff{useless} step is a $\ls$-step that is not useful. A
  \deff{useful derivation} (resp. \deff{useless derivation}) is a
  derivation whose steps are useful (resp. useless).
\end{definition}
Let us give some examples. The steps 
\begin{align*}
  (\ap\tm\var)\esub\var{\ap{(\l\vartwo.\vartwo)}\tmtwo}
  &\tols
  (\ap\tm{(\ap{(\l\vartwo.\vartwo)}\tmtwo)})\esub\var{\ap{(\l\vartwo.\vartwo)}\tmtwo};\\
  (\ap\var\tm)\esub\var{\l\vartwo.\vartwo}
  &\tols
  (\ap{(\l\vartwo.\vartwo)}\tm)\esub\var{\l\vartwo.\vartwo};
\end{align*}
are useful because they move or create a $\beta/\db$-redex (first and
second case of the definition, respectively) while
$$
(\l\var.\vartwo)\esub\vartwo{\ap\varthree\varthree}\tols(\l\var.(\ap\varthree\varthree))\esub\vartwo{\ap\varthree\varthree}
$$
is useless. However, useful steps are subtler, for instance 
$$
(\ap\tm\var)\esub\var{\ap\varthree\varthree}\esub\varthree{\l\vartwo.\vartwo}
\tols
(\ap\tm{(\ap\varthree\varthree)})\esub\var{\ap\varthree\varthree}\esub\varthree{\l\vartwo.\vartwo}
$$
is useful also if it does not move or create $\beta/\db$-redexes,
because it does so up to relative unfolding, \ie\
$\relunf{(\ap\varthree\varthree)}{\ctxhole\esub\varthree{\l\vartwo.\vartwo}}=\ap{(\l\vartwo.\vartwo)}{\l\vartwo.\vartwo}$
that is a $\beta/\db$-redex.

Note that \opt\ steps concern future creations of $\beta$-redexes and
yet their definition circumvent the explicit use of residuals, relying
on relative unfoldings only.

\paragraph{Leftmost-Outermost Useful Derivations.} 
The notion of small-step evaluation that we will use to implement \lo\
$\beta$-reduction is the one of \lo\ useful derivation. We need some
preliminary definitions.

Let $\redex$ be a redex. Its \deff{position} is defined as follows:
\begin{varenumerate}
\item 
  If $\redex$ is a $\db$-redex $\ctxp{\sctxp{\l \var.\tm} \tmtwo}
  \todb \ctxp{ \sctxp{\tm\esub{\var}{\tmtwo}}}$ then its position is
  given by the context $\ctx$ surrounding the changing expression;
  $\beta$-redexes are treated as $\db$-redexes.
\item 
  If $\redex$ is a $\ls$-redex, expressed in compact form
  $\ctxp{\var} \tols \ctxp{\tmtwo}$, then its position is the context
  $\ctx$ surrounding the variable occurrence to substitute.
\end{varenumerate}
The left-to-right outside-in order on redexes is expressed as an order
on positions, \ie\ contexts. Let us warn the reader about a possible
source of confusion. The \emph{left-to-right outside-in} order in the
next definition is sometimes simply called \emph{left-to-right} (or
simply \emph{left}) order. The former terminology is used when terms
are seen as trees (where the left-to-right and the outside-in orders
are disjoint), while the latter terminology is used when terms are
seen as strings (where the left-to-right is a total order). While the
study of standardisation for the \lsc~\cite{non-standard-preprint}
uses the string approach (and thus only talks about the
\emph{left-to-right} order and the \emph{leftmost} redex), here some
of the proofs require
a delicate analysis of the relative positions of redexes and so we
prefer the more informative tree approach and define the order
formally.

\begin{definition}
  The following definitions are given with respect to general (not
  necessarily shallow) contexts, even if apart from the next section
  we will use them only for shallow contexts.
  \begin{varenumerate}
  \item 
    The \deff{outside-in order}: 
    \begin{varenumerate}
    \item 
      \emph{Root}: $\ctxhole\outin\gctx$ for every context
      $\gctx\neq\ctxhole$;
    \item 
      \emph{Contextual closure}: If $\gctx\outin\gctxtwo$ then
      $\gctxthreep\gctx\outin\gctxthreep\gctxtwo$ for any context
      $\gctxthree$.
    \end{varenumerate}
    Note that $\outin$ can be seen as the prefix relation $\prefix$ on
      contexts.
  \item 
    The \deff{left-to-right order}: $\gctx\leftright\gctxtwo$ is defined by:
    \begin{varenumerate}
    \item 
      \emph{Application}: If $\gctx\prefix \tm$ and
      $\gctxtwo\prefix\tmtwo$ then
      $\ap\gctx\tmtwo\leftright\ap\tm\gctxtwo$;
    \item 
      \emph{Substitution}: If $\gctx\prefix \tm$ and
      $\gctxtwo\prefix\tmtwo$ then
      $\gctx\esub\var\tmtwo\leftright\tm\esub\var\gctxtwo$;
    \item 
      \emph{Contextual closure}: If $\gctx\leftright\gctxtwo$ then
      $\gctxthreep\gctx\leftright\gctxthreep\gctxtwo$ for any context
      $\gctxthree$.
    \end{varenumerate}  
  \item 
    The \deff{left-to-right outside-in order}:
    $\gctx\leftout\gctxtwo$ if $\gctx\outin\gctxtwo$ or
    $\gctx\leftright\gctxtwo$:
  \end{varenumerate}
\end{definition}

The following are a few examples. For every context $\gctx$, it holds
that $\ctxhole\not\leftright\gctx$. Moreover,
\begin{align*}
  \ap{(\l\var.\ctxhole)}\tm
  &\outin 
  \ap{(\l\var.(\ap{\ctxhole\esub\vartwo\tmtwo)}\tmthree)}\tm;\\
  \ap{(\ap\ctxhole\tm)}\tmtwo
  &\leftright
  \ap{(\ap\tmthree\tm)}\ctxhole;\\	
  \ap{\tm\esub\var\ctxhole}\tmtwo
  &\leftright
  \ap{\tm\esub\var\tmthree}\ctxhole.
\end{align*}
The next lemma guarantees that we defined a total order.
\begin{lemma}[Totality of $\leftout$]\label{l:lefttor-basic} 
  If $\gctx\prefix\tm$ and $\gctxtwo\prefix\tm$ then either
  $\gctx\leftout\gctxtwo$ or $\gctxtwo\leftout\gctx$ or
  $\gctx=\gctxtwo$.
\end{lemma}

\begin{proof}
  By induction on $\tm$. In general we can avoid to analyse the cases
  where $\gctx$ or $\gctxtwo$ is empty, because
  $\ctxhole\outin\gctxthree$ (and so $\ctxhole\leftout\gctxthree$) for
  any context $\gctxthree$ and there is no context
  $\gctxthree\neq\ctxhole$ s.t. $\gctxthree\leftout\ctxhole$. Cases:
  \begin{varenumerate}
  \item 
    \emph{Variable}. Both $\gctx$ and $\gctxtwo$ are the empty
    context, and so $\gctx=\gctxtwo$.
  \item 
    \emph{Abstraction $\l\var.\tmtwo$}. It follows from the \ih.
  \item 
    \emph{Application $\tmtwo\tmthree$}. If both contexts have
    their hole in $\tmtwo$ or both in $\tmthree$ we conclude by the
    \ih. Otherwise $\gctx$ has its hole in, say, $\tmtwo$ and
    $\gctxtwo$ in $\tmthree$, and then $\gctx\leftright\gctxtwo$, \ie\
    $\gctx\leftout\gctxtwo$.
  \item 
    \emph{Substitution $\tmtwo\esub\var\tmthree$}. Exactly as the
    previous case.
  \end{varenumerate}
\end{proof}

The orders above can be extended from contexts to redexes, in the
expected way, \eg\ for $\leftout$ given two redexes $\redex$ and
$\redextwo$ of positions $\ctx$ and $\ctxtwo$ we write
$\redex\leftout\redextwo$ if $\ctx\leftout\ctxtwo$. Now, we can define
the notions of derivations we are interested in.

\begin{definition}[Leftmost-Outermost (Useful) Redex]
  Let $\tm$ be a term and $\redex$ a redex of $\tm$. $\redex$ is the
  \deff{leftmost-outermost} (resp. \deff{leftmost-outermost useful},
  \lou\ for short) redex of $\tm$ if $\redex\leftout\redextwo$ for
  every other redex (resp. useful redex) $\redextwo$ of $\tm$. We
  write $\tm\tolo\tmtwo$ (resp. $\tm\tolou\tmtwo$) if a step reduces
  the \lo\ (resp. \lou) redex.
\end{definition}

We need to ensure that \lou\ derivations are mechanisable and form a
high-level implementation system when paired with \lo\ derivations. In
particular, we will show:
\begin{varenumerate}
\item 
  the \emph{subterm} and \emph{trace properties}, by first showing
  that they hold for every standard derivation, in \refsect{standard},
  and then showing that \lou\ derivations are standard, in
  \refsect{subterm-via-standard};
\item 
  the \emph{normal form} and \emph{projection properties}, by a
  careful study of unfoldings and \lo/\lou\ derivations, in
  \refsect{projection};
\item 
  the \emph{syntactic bound property}, passing through the
  abstract notion of \emph{nested derivation}, in \refsect{nested};
\item 
  the \emph{selection property}, by exhibiting a polynomial
  algorithm to test whether a redex is useful or not, in
  \refsect{algorithm}.
\end{varenumerate}

% !TEX root = main.tex
\section{Standard Derivations}
\label{sect:standard}
We need to show that \lou\ derivations have the subterm property. It
could be done directly. However, we will proceed in an abstract way,
by first showing that the subterm property is a property of standard
derivations for the LSC, and then showing (in
\refsect{subterm-via-standard}) that \lou\ derivations are
standard. The detour has the purpose of sheding a new light on the
notion of standard derivation, a classic concept in rewriting
theory. For the sake of readability, we use the concept of residual
without formally defining it (see \cite{non-standard-preprint} for
details).

\begin{definition}[Standard Derivation]
A derivation $\deriv:\redex_1;\ldots;\redex_n$ is \deff{standard} if
$\redex_i$ is not the residual of a redex $\redextwo\leftout\redex_j$
for every $i\in\set{2,\ldots,n}$ and $j<i$.
\end{definition}
The same definition where terms are ordinary $\l$-terms gives the
ordinary notion of standard derivation.

Note that any single reduction step is standard. Then,
notice that standard derivations select redexes in a left-to-right and
outside-in way, but they are not necessarily \lo. For instance, the
derivation
$$(\ap{(\l\var.\vartwo)}\vartwo)\esub\vartwo\varthree
\tols
(\ap{(\l\var.\varthree)}\vartwo)\esub\vartwo\varthree
\tols
(\ap{(\l\var.\varthree)}\varthree)\esub\vartwo\varthree$$
is standard even if the \lo\ redex (\ie\ the $\db$-redex on $\var$) is not reduced. The extension of the derivation with $(\ap{(\l\var.\varthree)}\varthree)\esub\vartwo\varthree\todb \varthree\esub\var\varthree\esub\vartwo\varthree$ is not standard. Last, note that the position of a $\ls$-step is given by the substituted occurrence and not by the ES, that is $(\ap\var\vartwo)\esub\var\tmtwo\esub\vartwo\tm
\tols
(\ap\var\tm)\esub\var\tmtwo\esub\vartwo\tm
\tols
(\ap\tmtwo\tm)\esub\var\tmtwo\esub\vartwo\tm$ is not standard.

In \cite{non-standard-preprint} it is showed that in the full
\lsc\ standard derivations are complete, \ie\ that whenever
$\tm\to^*\tmtwo$ there is a standard derivation from $\tm$ to
$\tmtwo$. The shallow fragment does not enjoy such a standardisation
theorem, as the residuals of a shallow redex need not be shallow. This
fact however does not clash with the technical treatment in this
paper.  The shallow restriction is indeed compatible with
standardisation in the sense that:
\begin{varenumerate}
\item 
  \emph{The linear \lo\ strategy is shallow}: if the initial term is a
  $\l$-term then every redex reduced by the linear \lo\ strategy is
  shallow (every non-shallow redex $\redex$ is contained in a
  substitution, and every substitution is involved in an outer
  redex $\redextwo$);
\item 
  \emph{$\leftout$-ordered shallow derivations are standard}: any
  strategy picking shallow redexes in a left-to-right and outside-in
  fashion does produce standard derivations (it follows from the easy
  fact that a shallow redex $\redex$ cannot turn a non-shallow redex
  $\redextwo$ s.t. $\redextwo\leftout\redex$ into a shallow redex).
\end{varenumerate}
Moreover, the only redex swaps we will consider (\reflemma{useless-pers}) will produce shallow residuals.

We are now going to show a fundamental property of standard
derivations. The \emph{subterm property} states that at any point of a
derivation $\deriv:\tm\to^*\tmtwo$ only sub-terms of the initial term
$\tm$ are duplicated. It immediately implies that any rewriting step
can be implemented in time polynomial in the size $\size\tm$ of
$\tm$. A first consequence is the fact that $\size\tmtwo$ is linear in
the size of the starting term and the number of steps, that we call
the \emph{no size explosion} property.

These properties are based on a technical lemma relying on the notions
of box context and box subterm, where a \emph{box} is the argument of
an application or the content of an explicit substitution,
corresponding to explicit boxes for promotions in the
\pns\ representation of $\l$-terms with ES.

\begin{definition}[Box Context, Box Subterm]
  Let $\tm$ be a term. \deff{Box contexts} (that are not necessarily
  shallow) are defined by the following grammar, where $\gctx$ is an
  generic context:
  \begin{center}
    $\begin{array}{rcl} 
      \bctx &\grameq& \ap\tm\ctxhole\midd
      \tm\esub\var\ctxhole\midd \gctxp\bctx.
    \end{array}$
  \end{center}
  A \deff{box subterm} of $\tm$ is a term $\tmtwo$
  s.t. $\tm=\bctxp\tmtwo$ for some box context $\bctx$.
\end{definition}

We are now ready for the lemma stating the fundamental invariant of standard derivations.

\begin{lemma}[Standard Derivations Preserve Boxes on Their Right]\label{l:subterm-aux}
  Let $\deriv:\tm_0\to^k\tm_k\to\tm_{k+1}$ be a standard derivation
  and let $\ctx$ be the position of the last contracted redex, $k\geq
  0$, and $\bctx\prefix\tm_{k+1}$ be a box context
  s.t. $\ctx\leftout\bctx$. Then the box subterm $\tmtwo$ identified
  by $\bctx$ (\ie\ s.t. $\tm_{k+1}=\bctxp{\tmtwo}$) is a box subterm
  of $\tm_0$.
\end{lemma}

\begin{proof}
	By induction on $k$. If $k=0$ the statement trivially holds. Otherwise, let us call $\redex_k$ the step $\tm_k\to\tm_{k+1}$ and consider the last contracted redex $\redex_{k-1}:\tm_{k-1}\to\tm_k$ of position $\ctxtwo_{k-1}$. By \ih\ the statement holds wrt box contexts $\bctx\prefix\tm_k$ s.t. $\ctxtwo_{k-1}\leftout\bctx$. The proof analyses the position $\ctxtwo_k$ of $\redex_k$ with respect to the position $\ctxtwo_{k-1}$ of $\redex_{k-1}$, often distinguishing between the left-to-right $\leftright$ and the outside-in $\outin$ suborders of $\leftout$. Cases:
	\begin{varenumerate}
		\item \emph{$\ctxtwo_k$ is equal to $\ctxtwo_{k-1}$}. Clearly if $\ctxtwo_{k-1}=\ctxtwo_k\leftright\bctx$ then the statement holds because of the \ih\ (reduction clearly does not affect boxes on the left of the hole of the position). We need to check the box contexts s.t. $\ctxtwo_k\outin\bctx$. Note that  $\redex_{k-1}$ cannot be a $\todb$ redex, because if $\tm_{k-1}=\ctxtwo_{k-1}\ctxholep{\sctxp{\l\var.\tmthree}\tmfour} \todb \ctxtwo_{k-1}\ctxholep{\sctxp{\tmthree\esub\var\tmfour}}=\tm_k$ then $\ctxtwo_{k-1}=\ctxtwo_k$ is not the position of any redex in $\tm_k$. Hence, $\redex_{k-1}$ is a $\tols$ step and there are two cases. If it substitutes a:
		\begin{varenumerate}
			\item \emph{Variable}, \ie\ the sequence of steps $\redex_{k-1};\redex_k$ is
			$$\tm_{k-1}=\ctxtwo_{k-1}\ctxholep\var\tols\ctxtwo_{k-1}\ctxholep\vartwo\tols\ctxtwo_{k-1}\ctxholep\tmthree=\tm_{k+1}$$
			Then all box subterms of $\tmthree$ come from box subterms that also appear on the left of $\ctxtwo_{k-1}$ in $\tm_k$ and so they are box subterms of $\tm_0$ by \ih.
			\item \emph{$\db$-redex}, \ie\ the sequence of steps $\redex_{k-1};\redex_k$ is
			$$\tm_{k-1}=\ctxtwo_{k-1}\ctxholep\var\tols\ctxtwo_{k-1}\ctxholep{\sctxp{\l\vartwo.\tmthree}\tmfour} \tols \ctxtwo_{k-1}\ctxholep{\sctxp{\tmthree\esub\vartwo\tmfour}}=\tm_{k+1}$$
			Then all box subterms of $\sctxp{\tmthree\esub\vartwo\tmfour}$ come from box subterms of $\sctxp{\l\vartwo.\tmthree}\tmfour$ and so they are box subterms of $\tm_0$ by \ih.
		\end{varenumerate}
		
		\item \emph{$\ctxtwo_{k-1}$ is internal to $\ctxtwo_k$}, \ie\ $\ctxtwo_k\outin\ctxtwo_{k-1}$ and $\ctxtwo_k\neq\ctxtwo_{k-1}$. This case is only possible if $\redex_k$ has been created \emph{upwards} by $\redex_{k-1}$, otherwise the derivation would not be $\leftout$-standard.	 There are only two possible cases of creations \emph{upwards}:
		\begin{varenumerate}
			\item $\db$ creates $\db$, \ie\ $\redex_{k-1}$ is 
			$$\tm_{k-1}=\ctxtwo_{k-1}\ctxholep{\sctxp{\l\vartwo.\sctxtwop{\l\varthree.\tmthree}}\tmfour}\todb\ctxtwo_{k-1}\ctxholep{\sctxp{\sctxtwop{\l\varthree.\tmthree}\esub\vartwo\tmfour}}=\tm_k$$
			and $\ctxtwo_{k-1}$ is applicative, that is $\ctxtwo_{k-1}=\ctxtwo_k\ctxholep{\sctxthreep\cdot\tmfive}$, so that $\redex_k$ is 
			$$\tm_k=\ctxtwo_k\ctxholep{\sctxthreep{\sctxp{\sctxtwop{\l\varthree.\tmthree}\esub\vartwo\tmfour}}\tmfive} \todb \ctxtwo_k\ctxholep{\sctxthreep{\sctxp{\sctxtwop{\tmthree\esub\varthree\tmfive}\esub\vartwo\tmfour}}}=\tm_{k+1}$$
			The box subterms of $\sctxthree$ and $\tmthree$ (included) are box subterms of $\tm_k$ whose box context $\bctx$ is $\ctxtwo_{k-1}\leftright\bctx$ and so they are box subterms of $\tm_0$ by \ih. The other box subterms of $\sctxthreep{\sctxp{\sctxtwop{\tmthree\esub\varthree\tmfive}\esub\vartwo\tmfour}}$ are instead box subterms of $\sctxp{\l\vartwo.\sctxtwop{\l\varthree.\tmthree}}\tmfour$ and so they are box subterms of $\tm_0$ by \ih.
			
			\item $\lssym$ creates $\db$, \ie\ $\redex_{k-1}$ is 
			$$\tm_{k-1}=\ctxtwo_{k-1}\ctxholep{\var}\todb\ctxtwo_{k-1}\ctxholep{\sctxp{\l\vartwo.\tmthree}}=\tm_k$$
			and $\ctxtwo_{k-1}$ is applicative, \ie\ $\ctxtwo_{k-1}=\ctxtwo_k\ctxholep{\sctxtwop\cdot\tmfour}$ so that $\redex_k$ is 
			$$\ctxtwo_k\ctxholep{\sctxtwop{\sctxp{\l\vartwo.\tmthree}}\tmfour} \todb \ctxtwo_k\ctxholep{\sctxtwop{\sctxp{\tmthree\esub\vartwo\tmfour}}}$$
			The box subterms of $\sctxtwo$ and $\tmfour$ (included) are box subterms of the ending term of $\redex_{k-1}$ whose box context $\bctx$ is $\ctxtwo_{k-1}\leftright\bctx$ and so they are box subterms of $\tm_0$ by \ih. The other box subterms of $\sctxtwop{\sctxp{\tmthree\esub\vartwo\tmfour}}$ are also box subterms of $\sctxp{\l\vartwo.\sctxtwop{\l\varthree.\tmthree}}\tmfour$ and so they are box subterms of $\tm_0$ by \ih.
		\end{varenumerate}

		\item \emph{$\ctxtwo_k$ is internal to $\ctxtwo_{k-1}$}, \ie\ $\ctxtwo_{k-1}\outin\ctxtwo_k$ and $\ctxtwo_k\neq\ctxtwo_{k-1}$. Cases of $\redex_{k-1}$:
		\begin{varenumerate}
			\item \label{p:subterm}\emph{$\db$-step}, \ie\ $\redex_{k-1}$ is 
			$$\tm_{k-1}= \ctxtwo_{k-1}\ctxholep{\sctxp{\l\var.\tmthree}\tmfour} \todb \ctxtwo_{k-1}\ctxholep{\sctxp{\tmthree\esub\var\tmfour}}=\tm_k$$
			Then the hole of $\ctxtwo_k$ is necessarily inside $\tmthree$. Box subterms identified by a box context $\bctx$ s.t. $\ctxtwo_k\leftright\bctx$ in $\tm_{k+1}$ are also box subterms of $\tm_k$, and so we conclude applying the \ih. For box subterms identified by a $\bctx$ of $\tm_{k+1}$ s.t. $\ctxtwo_k\outin\bctx$ we have to analyse $\redex_k$. Suppose that $\redex_k$ is a:
			\begin{varitemize}
				\item \emph{$\db$-step}. Note that in a root $\db$-step (\ie\ at top-level) all the box subterms of the reduct are box subterms of the redex. In this case the redex is contained in $\tmthree$ and so by \ih\ all such box subterms are box subterms of $\tm_0$. 
				
				\item \emph{$\lssym$-step}, \ie\ $\redex_k$ has the form $\tm_k=\ctxtwo_k\ctxholep{\var} \tols \ctxtwo_k\ctxholep{\tmfive}=\tm_{k+1}$. In $\tm_k$ $\tmfive$ is identified by a box context $\bctx$ s.t. $\ctxtwo_k\leftright\bctx$. From $\ctxtwo_{k-1}\leftout\ctxtwo_k$ we obtain $\ctxtwo_{k-1}\leftright\bctx$ and so all box subterms of $\tmfive$ are box subterms of $\tm_0$ by \ih.
			\end{varitemize}

			\item $\lssym$-step: $\redex_{k-1}$ is $\tm_{k-1}=\ctxtwo_{k-1}\ctxholep{\var} \todb \ctxtwo_{k-1}\ctxholep{\tmfive}=\tm_k$. It is analogous to the $\db$-case: $\redex_k$ takes place inside $\tmfive$, whose box subterms are box subterms of $\tm_0$, by \ih. If $\redex_k$ is a $\db$-redex then it only rearranges constructors in $\tmfive$ without changing box subterms, otherwise it substitutes something coming from a substitution that is on the left of $\ctxtwo_{k-1}$ and so whose box subterms are box subterms of $\tm_0$ by \ih.
		\end{varenumerate}

			\item \emph{$\ctxtwo_k$ is on the left of $\ctxtwo_{k-1}$}, \ie\ $\ctxtwo_{k-1}\leftright\ctxtwo_k$. So for $\bctx$ s.t. $\ctxtwo_k\leftright\bctx$ we conclude immediately using the \ih, because there is a box context $\bctxtwo$ in $\tm_k$ s.t. $\ctxtwo_{k-1}\leftright\bctxtwo$ and identifying the same box subterm of $\bctx$. For $\bctx$ s.t. $\ctxtwo_k\outin\bctx$ we reason as in case \ref{p:subterm}.		
	\end{varenumerate}
\end{proof}

From the invariant, one easily obtains the subterm property, that in
turn implies the no size explosion and the trace properties.

\begin{corollary}
	\label{coro:subterm}% \refcoro{subterm}.\ref{p:subterm-4}
	Let $\deriv:\tm\to^k\tmtwo$ be a standard derivation.
	\begin{varenumerate}
		\item \label{p:subterm-1}\emph{Subterm}: every $\tols$-step in $\deriv$ duplicates a subterm of $\tm$.

		\item \label{p:subterm-2}\emph{No Size Explosion}: $\size{\tmtwo}\leq (k+1)\cdot\size{\tm}$.
		
		\item \label{p:subterm-3}\emph{Trace}: if $\tm$ is an ordinary $\l$-term then $\esmeas\tmtwo=\sizedb\deriv$.
	\end{varenumerate}
\end{corollary}

The subterm property of standard derivations is specific to small-step evaluation at a distance, and it is the crucial half of the notion of mechanisable strategy. It allows to see the standardisation theorem as the unveiling of a \emph{very} abstract machine, hidden inside the calculus itself.

\begin{proof}
	\begin{varenumerate}
		\item By induction on $k$. If $k=0$ the statement is evidently true. Otherwise, by \ih\ for every $\deriv:\tm\to^{k-1}\tmthree$ all its $\tols$-steps duplicated subterms of $\tm$. If the next step is a $\db$-step we conclude, otherwise by \reflemma{subterm-aux} it duplicates a subterm of $\tmtwo$ which is a box subterm, and so a subterm, of $\tm$.
		
		\item It follows immediately from the subterm property.
		
		\item If $\tm$ is an ordinary $\l$-term then by the subterm property only ordinary $\l$-terms are duplicated. In particular, no explicit substitution constructor is ever duplicated by $\lssym$-steps in $\deriv$: if $\tmthree\tols\tmfour$ is a step of $\deriv$ then $\esmeas{\tmthree}=\esmeas{\tmfour}$. Every $\db$-step, instead, introduces a substitution, \ie\ $\esmeas\tmtwo=\sizedb\deriv$. 
		
	\end{varenumerate}
\end{proof}

Let us conclude the section with a further invariant of standard derivations. It is not needed for the invariance result, but sheds some light on the shallow subsystem under. Let a term be
\deff{shallow} if its substitutions do not contain substitutions. The
invariant is that if the initial term is a $\l$-term then standard shallow derivations involve only shallow terms. This fact is the only point of this section relying on the assumption that reduction is shallow (the standard hypothesis is also necessary, consider $(\l\var.\var) ((\l\vartwo.\vartwo)\varthree)\todb
(\l\var.\var)
(\vartwo\esub\vartwo\varthree)\todb\var\esub\var{\vartwo\esub\vartwo\varthree}$).

\begin{lemma}[Shallow Invariant]
Let $\tm$ be a $\l$-term and $\deriv:\tm\to^k\tmtwo$ be a standard derivation. Then $\tmtwo$ is a shallow term.
\end{lemma}

\begin{proof}
By induction on $k$. If $k=0$ the statement is evidently true. Otherwise, by \ih\ every explicit substitution in $\tmthree$, where $\deriv:\tm\to^{k-1}\tmthree$, contains a $\l$-term. If the next step is a:
			\begin{varenumerate}
				\item \emph{$\ls$-step}) By the subterm property the step duplicates a term without substitutions, and --- since reduction is shallow --- it does not put the duplicated term in a substitution. Therefore, for every substitution of $\tmtwo$ there is a substitution of $\tmthree$ with the same content. We then conclude applying the \ih.
				\item \emph{$\db$-step}) It is easily seen that the argument of the $\db$-step is on the right of the previous step, so that by \reflemma{subterm-aux} it contains a (box) subterm of $\tm$. Then, the substitution created by the $\db$-step contains a subterm of $\tm$, that is an ordinary $\l$-term by hypothesis. The step does not affect any other substitution, because reduction is shallow, and so we conclude.
			\end{varenumerate}
\end{proof}

% !TEX root = main.tex
%%%%%%%%%%%%%%%%%%%%%%%%%%%%%%%%%%%%%%%%%%%%%%%%%%%%%%%%%%%%%%%%%%%%%%%%%%%%%
\section{The Subterm and Trace Properties, via Standard Derivations}
%%%%%%%%%%%%%%%%%%%%%%%%%%%%%%%%%%%%%%%%%%%%%%%%%%%%%%%%%%%%%%%%%%%%%%%%%%%%%
\label{sect:subterm-via-standard}
While \lo\ derivations are evidently standard, a priori \lou\ derivations may not be standard, if the reduction of a useful redex $\redex$ could turn a useless redex  $\redextwo\leftout\redex$ into a useful redex. Luckily, this is not possible, \ie\ uselessness is stable by reduction of $\leftout$-majorants, as proved by the next lemma.

\begin{lemma}[Useless Persistence]
	\label{l:useless-pers}
	Let $\redex:\tm\tols\tmtwo$ be a useless redex and $\redextwo:\tm\to\tmthree$ be a useful redex s.t. $\redex\leftout\redextwo$. The unique residual $\redex'$ of $\redex$ after $\redextwo$ is shallow and useless.
%	Let $\redex:\tm\tols\tmtwo$ be a useless redex of position $\ctx$ and $\redextwo:\tm\to\tmthree$ be a useful redex of position $\ctxtwo$ and s.t. $\ctx\leftout\ctxtwo$. Then the unique residual $\redex'$ of $\redex$ after $\redextwo$ is useless.
\end{lemma}

\begin{proof}
	Let $\redex:\ctxtwop{\ctxp\var\esub\var\tmthree} \tols \ctxtwop{\ctxp\tmthree\esub\var\tmthree}$. Uselessness of $\redex$ implies that $\relunf{\tmthree}{\ctxtwo}$ is a normal $\l$-term and if $\relunf{\tmthree}{\ctxtwo}$ is an abstraction then $\ctxtwop{\ctx\esub\var\tmthree}$ is not an applicative context. Note that $\ls$-steps cannot change the useless nature of $\redex$, because 1) useful/useless redexes are defined via unfolding and 2) $\ls$-steps cannot change the applicative nature of a context (indeed that could happen if the hole is in a substitution and it is substituted on an applied variable occurrence, but here it is impossible because contexts are shallow). So, in the following we suppose that $\redextwo$ is a $\db$-redex. 
	
	By induction on $\ctxtwo$, the external context of $\redex$. Cases:
	\begin{varenumerate}
		\item \emph{Empty context $\ctxhole$}. Consider $\redextwo$, that necessarily takes the form 
		$$\redextwo: \ctxp\var\esub\var\tmthree \to \ctxtwop\var\esub\var\tmthree$$
		\ie\ reduction takes place in the context $\ctx$. The only way in which $\redex'$ can be useful is if $\redextwo$ turned the non-applicative context $\ctx$ into an applicative context $\ctxtwo$, assuming that $\unf{\tmthree}$ is an abstraction. The useful redex $\redextwo$ can change the nature of $\ctx$ only if $\ctx=\ctxthreep{\sctxp{\l\vartwo.\sctxtwo}\tmfour}$ and $\ctxthree$ is applicative, so that $\redextwo$ is 
			$$\ctxthreep{\sctxp{\l\vartwo.\sctxtwop\var}\tmfour}\esub\var\tmthree \todb
			\ctxthreep{\sctxp{\sctxtwop\var\esub\vartwo\tmfour}}\esub\var\tmthree $$
		with $\ctxtwo=\ctxthreep{\sctxp{\sctxtwo\esub\vartwo\tmfour}}$ applicative context. But then $\redextwo\leftout\redex$, against hypothesis. Absurd.
		
		\item \emph{Inductive cases}:
		\begin{varenumerate}
			\item \emph{Abstraction}, \ie\ $\ctxtwo=\l\vartwo.\ctxthree$. Both redexes $\redex$ and $\redextwo$ take place under the outermost abstraction, so we conclude using the \ih.
			
			\item \emph{Left of an application}, \ie\ $\ctxtwo=\ctxthree\tmfive$. Note that $\redextwo$ cannot be the eventual root $\db$-redex (\ie\ if $\ctxthree$ is of the form $\sctxp{\l\vartwo.\ctxfour}$ then $\redextwo$ is not the $\db$-redex involving $\l\vartwo$ and $\tmfive$), because this would contradict $\redex\leftout\redextwo$. If the redex $\redextwo$ takes place in $\ctxthreep{\ctxp\var\esub\var\tmthree}$ then we use the \ih. Otherwise $\redextwo$ takes place in $\tmfive$, the two redexes are disjoint, and commute. Evidently, the residual $\redex'$ of $\redex$ after $\redextwo$ is still shallow and useless.
			
			\item \emph{Right of an application}, \ie\ $\ctxtwo=\tmfive\ctxthree$. For the useless part it is similar to the previous case.  Now, $\redex'$ is shallow, because $\redex$ is shallow and $\redextwo$ cannot put its residual into an explicit substitution, because---as in the previous case---$\redextwo$ cannot be the eventual root $\db$-redex.
			
			\item \emph{Substitution}, \ie\ $\ctxtwo=\ctxthree\esub\vartwo\tmfive$. Both redexes $\redex$ and $\redextwo$ take place under the outermost explicit substitution $\esub\vartwo\tmfive$, so we conclude using the \ih.
		\end{varenumerate}
	\end{varenumerate}
\end{proof}

Using the lemma above and a technical property of standard derivations (the \emph{enclave axiom}, see \cite{non-standard-preprint}) we obtain:

\begin{proposition}[\lou-Derivations Are Standard]
\label{prop:lou-der-are-standard}
Let $\deriv$ be a \lou\ derivation. Then $\deriv$ is a standard derivation.
\end{proposition}

\begin{proof}
	Suppose not. Then $\deriv$ writes as $\derivtwo; \redex;\derivthree$ where $\derivtwo$ is the maximum standard prefix of $\deriv$ (that is necessarily non-empty, because a single step is always standard). Now, let $\derivtwo$ be $\redex_1;\ldots\redex_k$ and $\redex_i:\tm_{i}\to\tm_{i+1}$ with $i\in\set{1,\ldots,k}$. If $\derivtwo; \redex$ is not standard there is a term $\tm_i$ and a redex $\redextwo$ of $\tm_i$ s.t. 
	\begin{varenumerate}
		\item $\redex$ is a residual of $\redextwo$ after $\redex_i;\ldots;\redex_k$;
		\item $\redextwo\leftout\redex_i$.
	\end{varenumerate}
	By induction on $l=k-i+1$ (\ie\ on the length of the sequence $\redex_i;\ldots;\redex_k$) we prove that $\redex$ is useless. This fact will give us a contradiction, because by hypothesis $\redex$ is useful. 
	
	If $l=1$ by \reflemma{useless-pers} the unique residual $\redextwo_{i+1}$ of $\redextwo$ after $\redex_k$ is useless. But this residual, being unique, is necessarily equal to $\redex$, and we conclude. If $l>0$ consider $\redex_{i+1}$ and the unique residual $\redextwo_{i+1}$ of $\redextwo$ after $\redex_i$, which is useless by \reflemma{useless-pers}. Both $\redextwo_{i+1}$ and $\redex_{i+1}$ are redexes of $\tm_{i+1}$. Two cases:
	\begin{varenumerate}
		\item $\redex_{i+1}\leftout\redextwo_{i+1}$. The order $\leftout$ satisfies \mellies' axiomatics for standardisation, in particular the axiom \emph{Enclave} (see \cite{non-standard-preprint}, Sect. 9). We recall such an axiom, that will give us a contradiction. The axiom assumes two redexes $A$ and $B$ in a term $\tm$ s.t. $B\leftout A$ and so $B$ has a unique residual $B'$ after the reduction of $A$, and it has two parts:
		\begin{varenumerate}
			\item \emph{Creation}: if $A$ creates a redex $C$ then $B'\leftout C$;
			\item \emph{Nesting}: If $A\leftout C$ and $C'$ is a residual of $C$ after $A$ then $B'\leftout C'$
		\end{varenumerate}
Now, we have two cases:
		\begin{varenumerate}
			\item \emph{$\redex_{i+1}$ has been created by $\redex_i$}. Then by the creation part of the enclave axiom applied to $\redextwo\leftout\redex_i$ we obtain $\redextwo_{i+1}\leftout\redex_{i+1}$, absurd.
			\item \emph{$\redex_{i+1}$ is a residual after $\redex_i$ of a redex $\redex_{i+1}^{-1}$ of $\tm_i$}. Then by the nesting part of the enclave axiom applied to $\redextwo\leftout\redex_i$ we obtain $\redextwo_{i+1}\leftout\redex_{i+1}$, absurd.
		\end{varenumerate}
		\item $\redextwo_{i+1}\leftout\redex_{i+1}$. Then we can apply the \ih\ and conclude that $\redex$ is useless.
	\end{varenumerate}
\end{proof}

We conclude applying \refcoro{subterm}:
\begin{corollary}[Subterm and Trace]
\label{coro:subterm-for-lou}
\lou\ derivations have the subterm and the trace properties, and only involve shallow terms.
\end{corollary}

% !TEX root = main.tex
%%%%%%%%%%%%%%%%%%%%%%%%%%%%%%%%%%%%%%%%%%%%%%%%%%%%%%%%%%%
\section{The Normal Form and Projection Properties}
\label{sect:projection}
%%%%%%%%%%%%%%%%%%%%%%%%%%%%%%%%%%%%%%%%%%%%%%%%%%%%%%%%%%%
For the normal form property it is necessary to show that the position
of a redex in an unfolded term $\unf{\tm}$ can be traced back to the
position of a useful redex in the original term
$\tm$. 

  Given that unfoldings turn explicit into implicit
  substitutions, we need a few preliminary lemmas about the stability
  by substitutions of the notions of applicative contexts, outside-in
  order, and position. Remember that all contexts are implicitly assumed to be shallow.

\begin{lemma}[Applicative Contexts Are Stable by Substitution]
\label{l:apctx-stable}
	Let $\apctx$ be an applicative context. Then $\apctx\isub\var\tm$ is an applicative context.
\end{lemma}

\begin{proof}
	Straightforward induction on $\apctx$.
\end{proof}

\begin{lemma}[$\outin$ Is Stable by Substitution]
	\label{l:inst-and-sub}
	If $\ctx\outin\ctxtwo$ then $\ctx\isub\var\tm\outin\ctxtwo\isub\var\tm$.
\end{lemma}

\begin{proof}
By induction on $\ctx$.
\end{proof}

The definition of useful redex has a case about applicative contexts. The normal form property asks us to prove that the projection of a useful normal form is a $\beta$ normal form. Thus, we are forced to study how applicative contexts interact with substitution and unfoldings. This is why the following two lemmas have particularly technical statements and proofs.

\begin{lemma}[Positions and Substitution]
	\label{l:ctx-and-sub}
	Let $\tm$ and $\tmtwo$ be $\l$-terms and $\tm\isub\vartwo\tmtwo=\ctxp\var$. Then 
	\begin{varenumerate}
		\item there exists $\ctxtwo$ s.t. $\tm=\ctxtwop\varthree$ with $\varthree\in\set{\var,\vartwo}$ and $\ctxtwo\isub\vartwo\tmtwo\outin\ctx$.
		\item  \label{p:ctx-and-sub-app}if in addition $\ctx$ is applicative then either
		\begin{varenumerate}
			\item	\label{p:ctx-and-sub-app-1} $\ctxtwo$ is applicative and $\ctxtwo\isub\vartwo\tmtwo=\ctx$, or
			\item \label{p:ctx-and-sub-app-2} $\varthree=\vartwo$ and there is an applicative context $\ctxthree$ s.t. $\tmtwo=\ctxthreep\var$.
		\end{varenumerate}
	\end{varenumerate}
\end{lemma}

\begin{remark}
	\label{rem:compact-condition}
	In Point \ref{p:ctx-and-sub-app}.\ref{p:ctx-and-sub-app-1} asking that $\ctxtwo\isub\vartwo\tmtwo=\ctx$ is equivalent to:
	\begin{varenumerate}
		\item $\varthree\isub\vartwo\tmtwo=\var$, because both $\tm\isub\vartwo\tmtwo=\ctxtwop\varthree\isub\vartwo\tmtwo=\ctxtwo\isub\vartwo\tmtwo\ctxholep{\varthree\isub\vartwo\tmtwo}$ and $\tm\isub\vartwo\tmtwo=\ctxp\var$;
		\item $\varthree=\var$ or $\tmtwo=\var$, because it is equivalent to $\varthree\isub\vartwo\tmtwo=\var$.
	\end{varenumerate}
\end{remark}

\begin{proof}
By induction on $\tm$. Cases:
\begin{varitemize}
		\item \emph{Variable}. Sub-cases:
			\begin{varitemize}
				\item $\tm =\var$. 
				\begin{varenumerate}
					\item then $\ctx=\ctxhole$ and it is enough to take $\ctxtwo\defeq\ctxhole$, for which $\ctxhole\isub\vartwo\tmtwo=\ctxhole=\ctx$.
					\item $\ctx=\ctxhole$ is not applicative, so there is nothing to prove.
				\end{varenumerate}

				\item $\tm = \vartwo$.
				\begin{varenumerate}
					\item then $\ctx\prefix\tmtwo$ and $\ctxtwo\defeq\ctxhole$ verifies the statement (because $\ctxhole\outin\ctxthree$ for any context $\ctxthree$).
					\item If $\ctx$ is applicative then we are in case \ref{p:ctx-and-sub-app-2} and the applicative context $\ctxthree$ for $\tmtwo$ is $\ctx$ itself.
				\end{varenumerate}
				
				\item $\tm=\varthree$. Impossible because then $\tm\isub\vartwo\tmtwo=\varthree\neq\ctxp\var$.
			\end{varitemize}
		
		\item \emph{Abstraction}. It follows immediately from the \ih.
		
		\item \emph{Application $\tm=\tmthree\tmfour$}. Then $\tm\isub\vartwo\tmtwo=\tmthree\isub\vartwo\tmtwo\tmfour\isub\vartwo\tmtwo$. Two cases:
		\begin{varitemize}
			\item \emph{The hole of $\ctx$ is in the left subterm $\tmthree\isub\vartwo\tmtwo$}. Then $\ctx=\ctxthree\tmfour\isub\vartwo\tmtwo$ and so $\tmthree\isub\vartwo\tmtwo=\ctxthreep\var$. 
			\begin{varenumerate}
				\item By \ih\ we obtain $\ctxfour$ s.t. $\tmthree=\ctxfourp\varthree$ with $\varthree\in\set{\var,\vartwo}$ and $\ctxfour\isub\vartwo\tmtwo\outin\ctxthree$. Then $\ctxtwo\defeq\ctxfour\tmfour$ is s.t. $\tm=\ctxtwop\varthree$ and $\ctxtwo\isub\vartwo\tmtwo = \ctxfour\isub\vartwo\tmtwo \tmfour\isub\vartwo\tmtwo \outin \ctxthree \tmfour\isub\vartwo\tmtwo = \ctx$.

				\item If $\ctx$ is applicative there are two sub-cases:
					\begin{varenumerate}
						\item \emph{$\ctxthree$ is applicative}. Then we conclude using the \ih.
						\item \emph{$\ctxthree$ is not applicative}. Then $\ctxthree=\ctxhole$ and the context $\ctxfour$ given by the \ih\ is empty, \ie\ $\ctxtwo=\ctxhole\tmfour$ is applicative and verifies $\ctxtwo\isub\vartwo\tmtwo=\ctxhole\tmfour\isub\vartwo\tmtwo=\ctx$, \ie\ this is an instance of case \ref{p:ctx-and-sub-app-1}.
					\end{varenumerate}
			\end{varenumerate}
			
			\item \emph{The hole of $\ctx$ is in the right subterm $\tmfour\isub\vartwo\tmtwo$}. Then $\ctx=\tmthree\isub\vartwo\tmtwo\ctxthree$ and so $\tmfour\isub\vartwo\tmtwo=\ctxthreep\var$.
			\begin{varenumerate}
				\item Analogous to the previous case.
				\item If $\ctx$ is applicative then $\ctxthree$ is applicative and we use the \ih.
			\end{varenumerate}
		\end{varitemize}
	\end{varitemize}
\end{proof}

The next lemma states the technical relationship between positions and unfolding. It will be used in both the proofs of the normal form property (\refprop{opt-nf-to-nf}) and of the key lemma for the projection property (\reflemma{useful-projection}).

\begin{lemma}[Positions and Unfolding]
	\label{l:varoc-relunf-new} % 
	Let $\tm$ be a \lsc\ term and $\unf{\tm}=\ctxp\var$, where $\ctx$ does not capture $\var$. Then 
	\begin{varenumerate}
		\item there exists a shallow context $\ctxtwo$ s.t. $\tm=\ctxtwop\vartwo$ (possibly with $\vartwo=\var$) and $\unf{\ctxtwo}\outin\ctx$;
		\item if in addition $\ctx$ is applicative then either 
		\begin{varenumerate}
			\item \label{p:varoc-relunf-new-1} $\ctxtwo$ is applicative and $\unf{\ctxtwo}=\ctx$, or
			\item \label{p:varoc-relunf-new-2} there is an applicative context $\apctx$ s.t. $\relunf{\vartwo}{\ctxtwo}=\apctxp\var$.
		\end{varenumerate}
	\end{varenumerate}
\end{lemma}

\begin{remark}
\label{rem:compact-condition-2}
In \refpoint{varoc-relunf-new-1} $\unf{\ctxtwo}=\ctx$ is equivalent to $\relunf{\vartwo}{\ctxtwo}=\var$, because $\unf{\tm}=\unf{\ctxtwo}{\ctxholep{\relunf{\vartwo}{\ctxtwo}}}$ (by \reflemma{ctx-unf}.\ref{p:ctx-unf-three}) and $\unf{\tm}=\ctxp\var$ (by hypothesis).
\end{remark}

\begin{proof}
We prove both points together, by induction on $\tm$. Cases:
\begin{varitemize}
	\item \emph{Variable}. 
	\begin{varenumerate}
		\item If $\tm=\var$ then $\unf{\tm}=\var$, $\ctx=\ctxhole$ and it is enough to take $\ctxtwo\defeq\ctxhole$, for which $\unf{\ctxtwo}=\ctxhole\outin\ctxhole=\ctx$. 
		\item $\ctx$ cannot be applicative.
	\end{varenumerate}
	
	\item \emph{Abstraction}. It follows immediately by the \ih, since $\unf{\tm}$ is an abstraction.
	
	\item \emph{Application $\tm=\tmtwo\tmthree$}. By definition $\unf{\tm}=\unf{\tmtwo}\unf{\tmthree}$ and $\unf{\tm}=\ctxp\var$ implies $\ctx=\ctxthree\unf{\tmthree}$ or $\ctx=\unf{\tmtwo}\ctxthree$.
	\begin{varenumerate}
		\item  For both cases of $\ctx$ the \ih\ gives $\ctxtwo'$ satisfying the statement wrt $\ctxthreep\var$. We conclude taking $\ctxtwo\defeq\ctxtwo'\unf{\tmthree}$ or $\ctxtwo\defeq\unf{\tmtwo}\ctxtwo'$, depending on the case under consideration. 
		
		\item Suppose that $\ctx$ is applicative. If $\ctxthree$ is applicative then it follows from the \ih. Otherwise $\ctx=\ctxhole\unf{\tmthree}$, and that can happen only if $\unf{\tmtwo}=\var$, \ie\ if $\tmtwo=\sctxp\vartwo$ (possibly with $\vartwo=\var$) and $\relunf{\vartwo}{\sctx}=\var$. Then take $\ctxtwo\defeq\sctx\tmthree$, which is applicative and verifies $\tm=\ctxtwop\vartwo$ and $\unf{\ctxtwo}=\ctx$. We conclude because $\relunf{\vartwo}{\ctxtwo}=\relunf{\vartwo}{\sctx\tmthree}=\relunf{\vartwo}{\sctx}=\var$, and so we are in case \ref{p:varoc-relunf-new-1}.
	\end{varenumerate}
	
	\item \emph{Substitution $\tm=\tmtwo\esub\vartwo\tmthree$}. By definition $\unf{\tm}=\unf{\tmtwo}\isub\vartwo{\unf{\tmthree}}$.
	\begin{varenumerate}
		\item  By \reflemma{ctx-and-sub} there exists $\ctxthree$ s.t. $\unf{\tmtwo}=\ctxthreep{\varthree}$ with $\varthree\in\set{\vartwo,\var}$, and $\ctxthree\isub{\vartwo}{\unf{\tmthree}}\outin\ctx$. By \ih\ applied to $\tmtwo$, there also exists a context $\ctxfour$ s.t. $\tmtwo=\ctxfourp{\varthree'}$ (possibly with $\varthree'=\varthree$) and $\unf{\ctxfour}\outin\ctxthree$. Take $\ctxtwo\defeq\ctxfour\esub\vartwo\tmthree$. By definition $\unf{\ctxtwo}=\unf{\ctxfour}\isub\vartwo{\unf\tmthree}$. By \reflemma{inst-and-sub} applied to $\unf{\ctxfour}\outin\ctxthree$ we obtain $\unf{\ctxfour}\isub\vartwo{\unf\tmthree}\outin\ctxthree\isub\vartwo{\unf\tmthree}$, implying $\unf{\ctxtwo}\outin\ctxthree\isub\vartwo{\unf\tmthree}\outin\ctx$, that is $\unf{\ctxtwo}\outin\ctx$.
		
		\item If $\ctx$ is applicative then by Lemma \ref{l:ctx-and-sub} one of the two following cases applies:
	\begin{enumerate}
		\item \emph{$\ctxthree$ is applicative and $\ctxthree\isub{\vartwo}{\unf\tmthree}=\ctx$}. The \ih\ gives us one of the two following cases:
		\begin{varenumerate}
			\item \emph{$\ctxfour$ is applicative and $\unf{\ctxfour}=\ctxthree$}. Then $\ctxtwo$ (defined as $\ctxfour\esub\vartwo\tmthree$) is applicative and $\unf{\ctxtwo}=\unf{\ctxfour\esub\vartwo\tmthree}=\unf{\ctxfour}\isub\vartwo{\unf\tmthree}=\ctxthree\isub\vartwo{\unf\tmthree}=\ctx$. So we are in case \ref{p:varoc-relunf-new-1}.
			
			\item \emph{There is an applicative context $\apctxtwo$ s.t. $\relunf{\varthree'}{\ctxfour}=\apctxtwop{\varthree}$}. We have $\tm=\ctxtwop{\varthree'}$ so we want to show that there is an applicative context $\apctx$ s.t. $\relunf{\varthree'}{\ctxtwo}=\apctxp\var$, \ie\ that we are in case \ref{p:varoc-relunf-new-2}. Note that $\relunf{\varthree'}{\ctxtwo}=\relunf{\varthree'}{\ctxfour\esub\vartwo\tmthree}=\relunf{\varthree'}{\ctxfour}\isub\vartwo{\unf\tmthree}=\apctxtwop{\varthree}\isub\vartwo{\unf\tmthree}$. Two further sub-cases:
			\begin{varenumerate}
				\item $\varthree=\var$. Then $\relunf{\varthree'}{\ctxtwo}=\apctxtwop{\var}\isub\vartwo{\unf\tmthree}=\apctxtwo\isub\vartwo{\unf\tmthree}\ctxholep{\var}$. We conclude taking $\apctx\defeq\apctxtwo\isub\vartwo{\unf\tmthree}$, that is applicative by \reflemma{apctx-stable}.

				\item $\varthree=\vartwo$. Note that by $\unf{\ctxfour}=\ctxthree$ and \refremark{compact-condition} it follows that $\unf\tmthree=\var$. Then $\relunf{\varthree'}{\ctxtwo}=\apctxtwop{\vartwo}\isub\vartwo\var=\apctxtwo\isub\vartwo\var\ctxholep{\var}$. We conclude taking $\apctx\defeq\apctxtwo\isub\vartwo\var$, that is applicative by \reflemma{apctx-stable}.
			\end{varenumerate}			
		\end{varenumerate}
	
	\item \emph{$\varthree=\vartwo$ and there is an applicative context $\apctxtwo$ s.t. $\unf\tmthree=\apctxtwop\var$}. Then $\tm=\ctxtwop\vartwo= \ctxfourp\vartwo\esub\vartwo\tmthree$, and $\relunf{\vartwo}{\ctxtwo}=\relunf{\vartwo}{\ctxfour\esub\vartwo\tmthree}=\relunf{\vartwo}{\ctxfour}\isub\vartwo{\unf\tmthree}=\vartwo\isub\vartwo{\unf\tmthree}=\unf\tmthree=\apctxtwop\var$. So taking $\apctx\defeq\apctxtwo$ we are in case \ref{p:varoc-relunf-new-2}.
	
	\end{enumerate}	

	\end{varenumerate}
\end{varitemize}
\end{proof}

We now have the tools to prove the normal form property.

\begin{proposition}[Normal Form]
\label{prop:opt-nf-to-nf}
	Let $\tm$ be a \lsc\ term in useful normal form. Then $\unf{\tm}$ is a $\beta$-normal form.
\end{proposition}

\begin{proof}
By induction on $\tm$. Cases:
\begin{varenumerate}
	\item \emph{Variable $\tm=\var$}. Obvious.
	\item \emph{Abstraction $\tm=\l \var.\tmtwo$}. By \ih\ $\unf{\tmtwo}$ is a normal form. We conclude, since by definition $\unf{\tm}=\l\var.\unf{\tmtwo}$.
	\item \emph{Application $\tm=\tmtwo\tmthree$}. By hypothesis $\tmtwo$ cannot be an abstraction (otherwise $\tm$ would not be an \opt\ normal form). By definition $\unf{\tm}=\unf{\tmtwo}\unf{\tmthree}$ and by \ih\ $\unf{\tmtwo}$ and $\unf{\tmthree}$ are normal forms. We only need to show that $\unf{\tmtwo}$ cannot be an abstraction. Suppose it is. Then $\tmtwo$ has the form $\sctxp{\var}$ with $\sctx$ a substitution context acting on $\var$ and (\ie\ $\sctx=\sctxtwop{\sctxthree\esub\var\tmfour}$) s.t. $\relunf{\var}{\sctx}$ is an abstraction. Note that $\var$ occurs in an applicative context ($\sctx\tmthree$), and so the action of $\sctx$ on $\var$ is an \opt\ step, against the hypothesis that $\tm$ is an \opt\ normal form.
	
	\item \emph{Substitution $\tm=\tmtwo\esub\var\tmthree$}. By \ih\ $\unf{\tmtwo}$ is a  normal form. Then $\unf{\tm}=\unf{\tmtwo}\isub\var{\unf{\tmthree}}$ is not a normal form only if there is $\ctx$ s.t. $\unf{\tmtwo}=\ctxp\var$, and one of the two following cases holds:
	
	\begin{varenumerate}
		\item \label{p:optnf-unf-to-nf-sub-redex} \emph{$\unf\tmthree$ has a $\beta$-redex}. Then by \reflemma{varoc-relunf-new} applied to $\tmtwo$ there exists a context $\ctxtwo$ s.t. $\tmtwo=\ctxtwop\vartwo$ (possibly with $\vartwo=\var$) and $\unf{\ctxtwo}\outin\ctx$.  From $\unf{\tmtwo}=\unf{\ctxtwo} \ctxholep{\relunf\vartwo\ctxtwo}$ (by \reflemma{ctx-unf}.\ref{p:ctx-unf-three}) and $\unf{\ctxtwo}\outin\ctx$ we obtain that there exists $\ctxthree$ s.t. $\relunf\vartwo\ctxtwo=\ctxthreep\var$. So,
		\begin{varenumerate}
			\item $\tm=\ctxtwop\vartwo\esub\var\tmthree$, 
			\item the context $\ctxtwo\esub\var\tmthree$ can act on $\vartwo$ (also in the case that $\vartwo=\var$), and
			\item $\relunf\vartwo{\ctxtwo\esub\var\tmthree}=\relunf\vartwo{\ctxtwo}\isub\var{\unf\tmthree}=\ctxthreep\var\isub\var{\unf\tmthree}=\ctxthree\isub\var{\unf\tmthree}\ctxholep{\unf\tmthree}$ contains a $\beta$-redex. 
		\end{varenumerate}
		Then $\tm$ has a \opt\ redex, against hypothesis, absurd.
		
		\item \emph{$\unf\tmthree$ is an abstraction and $\ctx$ is applicative}. Then by \reflemma{varoc-relunf-new} applied to $\tmtwo$ there exists a context $\ctxtwo$ s.t. $\tmtwo=\ctxtwop\vartwo$ (possibly with $\vartwo=\var$), $\unf{\ctxtwo}\outin\ctx$, and either
		\begin{varenumerate}
			\item \emph{$\ctxtwo$ is applicative and $\unf{\ctxtwo}=\ctx$}. By \refremark{compact-condition-2} this is equivalent to $\relunf{\vartwo}{\ctxtwo}=\var$. Note that:
			\begin{varenumerate}
				\item $\tm=\ctxtwop\vartwo\esub\var\tmthree$,
				\item $\ctxtwo\esub\var\tmthree$ is applicative,
				\item $\relunf{\vartwo}{\ctxtwo\esub\var\tmthree}=\relunf{\vartwo}{\ctxtwo}\isub\var{\unf\tmthree}=\var\isub\var{\unf\tmthree}=\unf\tmthree$ is an abstraction.

			\end{varenumerate}
			Then $\tm$ has a \opt\ redex, against hypothesis, absurd.
			
			\item \emph{There exists an applicative context $\apctx$ s.t. $\relunf{\vartwo}{\ctxtwo}=\apctxp\var$}.
			\begin{varenumerate}
				\item $\tm=\ctxtwop\vartwo\esub\var\tmthree$,
				\item the context $\ctxtwo\esub\var\tmthree$ can act on $\vartwo$ (also in the case that $\vartwo=\var$), and
				\item $\relunf{\vartwo}{\ctxtwo\esub\var\tmthree}=\relunf{\vartwo}{\ctxtwo}\isub\var{\unf\tmthree}=\apctxp\var\isub\var{\unf\tmthree}=\apctxp{\unf\tmthree}$ is a $\beta$-redex, because $\unf\tmthree$ is an abstraction.
			\end{varenumerate}
			Then $\tm$ has a \opt\ redex, against hypothesis, absurd.				
			
		\end{varenumerate}
	\end{varenumerate}
\end{varenumerate}
\end{proof}

The next lemma shows that useful reductions match their intended semantics, in the sense that every useful redex contributes somehow to a $\beta$-redex. It is not needed for the invariance result.

\begin{lemma}[Inverse Normal Form]
\label{l:nf-to-opt-nf}
	Let $\tm$ be a \lsc\ term s.t. $\unf{\tm}$ is a $\beta$-normal form. Then $\tm$ is a useful normal form.
\end{lemma}

\begin{proof}
By contraposition. Suppose that $\tm\to\tmtwo$ by reducing a useful redex. A formal proof would be by induction on the useful step, however there are high-level observations that allow to avoid a formal proof. Note that since shallow $\db$-redexes cannot be erased by unfolding (because they take place in a shallow context) they project on $\beta$-redexes. Thus, we can assume that the step is a $\ls$-step. Now, since applicative contexts unfold to applicative contexts, it is evident by the definition of useful $\ls$-redex that a useful redex unfolds to a $\beta$-redex (that, as for $\db$-redexes, cannot be erased by the unfolding because its position is shallow). In any case $\unf\tm$ is not $\beta$-normal.
\end{proof}

% !TEX root = main.tex
For the projection property, we first need to show that the \lo\ order is stable by unfolding. As for positions, we first show that the \lo\ order is stable by substitution.

\begin{lemma}[$\leftout$ and Substitution]
\label{l:ctx-lefttor-sub}
Let $\tm$ be a $\l$-term, $\ctx\prefix\tm$ and $\ctxtwo\prefix\tm$. If $\ctx\isub\var\tmtwo\leftout\ctxtwo\isub\var\tmtwo$ then $\ctx\leftout\ctxtwo$.
\end{lemma}

\begin{proof}
	By induction on $\tm$. Cases:
	\begin{varenumerate}
		\item \emph{Variable}. If $\tm$ is a variable then both $\ctx$ and $\ctxtwo$ are the empty context, and so are the contexts $\ctx\isub\var\tmtwo$ and $\ctxtwo\isub\var\tmtwo$. Hence $\ctx\isub\var\tmtwo\not\leftout\ctxtwo\isub\var\tmtwo$ and the statement trivially holds.
		\item \emph{Abstraction $\l\var.\tmthree$}. It follows from the \ih.
		\item \emph{Application $\tmthree\tmfour$}. If $\ctx\isub\var\tmtwo$ is empty then $\ctx$ is empty and $\ctxtwo\isub\var\tmtwo$ is non-empty, that implies $\ctxtwo$ non-empty. Then $\ctx\leftout\ctxtwo$. If $\ctx\isub\var\tmtwo$ is non-empty then so are $\ctx$, $\ctxtwo\isub\var\tmtwo$, and $\ctxtwo$. We have $\tm\isub\var\tmtwo=\tmthree\isub\var\tmtwo\tmfour\isub\var\tmtwo$. Cases:
		\begin{varenumerate}
			\item \emph{$\ctx\isub\var\tmtwo$ and $\ctxtwo\isub\var\tmtwo$ both have their holes in $\tmthree\isub\var\tmtwo$}, \ie\ $\ctx\isub\var\tmtwo=\ctxthree\isub\var\tmtwo\tmfour\isub\var\tmtwo$ and $\ctxtwo\isub\var\tmtwo=\ctxfour\isub\var\tmtwo\tmfour\isub\var\tmtwo$ for some contexts $\ctxthree$ and $\ctxfour$. Since $\ctx\isub\var\tmtwo\leftout\ctxtwo\isub\var\tmtwo$ implies $\ctxthree\isub\var\tmtwo\leftout\ctxfour\isub\var\tmtwo$, by \ih\ we obtain $\ctxthree\leftout\ctxfour$ and so $\ctxthree\tmfour\leftout\ctxfour\tmfour$, \ie\ $\ctx\leftout\ctxtwo$.

			\item \emph{$\ctx\isub\var\tmtwo$ and $\ctxtwo\isub\var\tmtwo$ both have their holes in $\tmfour\isub\var\tmtwo$}. Analogous to the previous case.
			
			\item \emph{$\ctx\isub\var\tmtwo$ has its hole in $\tmthree\isub\var\tmtwo$ and $\ctxtwo\isub\var\tmtwo$ in $\tmfour\isub\var\tmtwo$}, \ie\ $\ctx\isub\var\tmtwo=\ctxthree\isub\var\tmtwo\tmfour\isub\var\tmtwo$ and $\ctxtwo\isub\var\tmtwo=\tmthree\isub\var\tmtwo\ctxfour\isub\var\tmtwo$ for some contexts $\ctxthree$ and $\ctxfour$. Then $\ctx=\ctxthree\tmfour$ and $\ctxtwo=\tmthree\ctxfour$, therefore $\ctx\leftout\ctxtwo$.
		\end{varenumerate}
	\end{varenumerate}
\end{proof}

\begin{lemma}[$\leftout$ and Unfolding]
\label{l:lefttor-prop} %\reflemma{lefttor-prop}.\ref{p:lefttor-prop-3}
	Let $\tm$ be a \lsc\ term, $\ctx\prefix\tm$ and $\ctxtwo\prefix\tm$. If $\unf{\ctx}\leftout\unf{\ctxtwo}$ then $\ctx\leftout\ctxtwo$.
\end{lemma}

\begin{proof}
 By induction on $\tm$. Note that $\unf{\ctxtwo}$ cannot be the empty context, because there is no context $\ctx$ s.t. $\ctx\leftout\ctxhole$. Moreover, if $\unf{\ctx}$ is the empty context then $\ctx$ is the empty context and the statement is immediately verified, because $\ctxhole\leftout\ctxtwo$ for all contexts $\ctxtwo$. Thus we can always exclude the cases where $\unf{\ctx}$ or $\unf{\ctxtwo}$ is empty. Cases of $\tm$:
	\begin{varenumerate}
		\item \emph{Variable $\var$}. We have $\unf{\var}=\var$ and a variable does not admit two different contexts, so there is nothing to prove.
		
		\item \emph{Abstraction $\l\var.\tmtwo$}. Then $\unf{(\l\var.\tmtwo)}=\l\var.\unf{\tmtwo}$. Then $\ctx=\l\var.\ctxthree$ and $\ctxtwo=\l\var.\ctxfour$, with $\ctxthree\leftout\ctxfour$, and we conclude using the \ih\ and the closure by contexts of $\leftout$.
		
		\item \emph{Application $\tmtwo\tmthree$}. Then $\unf{(\tmtwo\tmthree)}=\unf{\tmtwo} \unf{\tmthree}$. If $\unf{\ctx}$ and $\unf{\ctxtwo}$ both have their hole in $\unf{\tmtwo}$ or both in $\unf{\tmthree}$ then we conclude using the \ih\ and the closure by contexts of $\leftout$. Otherwise $\unf{\ctx}\prefix\unf{\tmtwo}$ and $\unf{\ctxtwo}\prefix\unf{\tmthree}$. Then necessarily $\ctx\prefix\tmtwo$ and $\ctxtwo\prefix\tmthree$, hence $\ctx\leftout\ctxtwo$.
		
		\item \emph{Substitution $\tmtwo\esub\var\tmthree$}. Then $\unf{\tm}=\unf{\tmtwo\esub\var\tmthree}=\unf{\tmtwo}\isub\var{\unf{\tmthree}}$. Since contexts are shallow, $\ctx=\ctxthree\esub\var\tmthree$ and $\ctxtwo=\ctxfour\esub\var\tmthree$ for some contexts $\ctxthree\prefix\tmtwo$ and $\ctxfour\prefix\tmtwo$. Then $\unf{\ctx}=\unf{\ctxthree}\isub\var{\unf{\tmthree}}$ and $\unf{\ctxtwo}=\unf{\ctxfour}\isub\var{\unf{\tmthree}}$ and the hypothesis becomes $\unf{\ctxthree}\isub\var{\unf{\tmthree}}\leftout \unf{\ctxfour}\isub\var{\unf{\tmthree}}$. \reflemma{ctx-lefttor-sub} gives $\unf{\ctxthree}\leftout \unf{\ctxfour}$. The \ih\ gives $\ctxthree\leftout \ctxfour$, that implies $\ctxthree\esub\var\tmthree\leftout \ctxfour\esub\var\tmthree$, \ie\ $\ctx\leftout\ctxtwo$.
	\end{varenumerate}
\end{proof}

We also need the two following straightforward properties.

\begin{lemma}
\label{l:beta-sub-left}
Let $\tm$ and $\tmtwo$ be $\l$-terms. If $\tm\rtob\tmthree$ then $\tm\isub\var\tmtwo\rtob\tmtwo\isub\var\tmtwo$.
\end{lemma}

\begin{proof}
If $\tm= (\l\vartwo.\tmfour)\tmfive\rtob\tmfour\isub{\vartwo}{\tmfive}=\tmthree$ then:
\[\begin{array}{llllllll}
\tm\isub\var\tmtwo&=&((\l\vartwo.\tmfour)\tmfive)\isub\var\tmtwo &=\\
&&(\l\vartwo.\tmfour\isub\var\tmtwo)\tmfive\isub\var\tmtwo&=\\
&&\tmfour\isub\var\tmtwo\isub{\vartwo}{\tmfive\isub\var\tmtwo}&=\\
&&\tmfour\isub{\vartwo}{\tmfive}\isub\var\tmtwo&=&\tmthree\isub\var\tmtwo
\end{array}\]
\end{proof}

\begin{lemma}
\label{l:rel-unf-and-sub-commute}
Let $\tm$ be a \lsc\ term and $\sctx$ be a substitution context. Then $\relunf{\tm}{\ctxp{\sctx}}=\relunf{\sctxp\tm}{\ctx}$
\end{lemma}

\begin{proof}
By induction on $\ctx$.
\end{proof}

We now dispose of all the ingredients for the proof of the key lemma on which the projection theorem relies on. We use $\rtob$ for $\beta$-reduction at top level.

\begin{lemma}[\lou\ $\db$-Step Projects on $\toblo$]
	\label{l:useful-projection}

	Let $\tm$ be a \lsc\ term and $\redex:\tm=\ctxp {\tmthree}\todb\ctxp{\tmfour}=\tmtwo$ with $\tmthree\rtodb\tmfour$. Then:
	\begin{varenumerate}

		\item \emph{Projection}: $\unf{\redex}:\unf{\tm}=\unf{\ctx}\ctxholep{\relunf{\tmthree}{\ctx}}\tob\unf{\ctx}\ctxholep{\relunf{\tmfour}{\ctx}}=\unf{\tmtwo}$ with $\relunf{\tmthree}{\ctx}\rtob\relunf{\tmfour}{\ctx}$;

		\item \emph{Minimality}: if moreover $\redex$ is the \lou\ redex in $\tm$ then $\unf{\redex}$ is the \lo\ $\beta$-redex in $\unf{\tm}$.
	\end{varenumerate}
\end{lemma}

The first point is an ordinary projection of reductions. The second one is instead involved, as it requires to prove that if $\unf{\redex}$ is not \lo\ then $\redex$ is not \lou, \ie\ to be able to somehow trace \lo\ redexes back through unfoldings. The proof is by induction on $\ctx$, that by hypothesis is the position of the \lou\ redex. The difficult case---not surprisingly---is when $\ctx=\ctxtwo\esub\var\tmfour$, and where the two lemmas relating the unfolding with positions (\reflemma{varoc-relunf-new}) and the order (\reflemma{lefttor-prop}) are applied. The proof also uses the normal form property, when the position $\ctx$ is on the argument $\tmfour$ of an application $\tmthree\tmfour$. Since $\redex$ is \lou, $\tmthree$ is useful-normal. To prove that $\unf{\redex}$ is the \lo\ $\beta$ redex in $\unf{(\tmthree\tmfour)}=\unf{\tmthree}\unf{\tmfour}$ we use the fact that $\unf{\tmthree}$ is normal.

\begin{proof}
By induction on $\ctx$. \emph{Notation}: if $\sctx=\esub{\var_1}{\tmthree_1}\ldots\esub{\var_n}{\tmthree_n}$ we denote with $\sctximp$ the list of implicit substitutions $\isub{\var_1}{\tmthree_1}\ldots\isub{\var_n}{\tmthree_n}$. Cases:
\begin{varenumerate}
	\item \emph{Empty context $\ctx=\ctxhole$}. Then $\tm=\tmp=\sctxp{\l\var.\tmthree}\tmfour$, $\unf{\tm}=\relunf{\tmp}{\ctx}$, $\tmtwo=\tmtwop$, $\unf{\tmtwo}=\relunf{\tmtwop}{\ctx}$, and
		\begin{center}
			\commDiagramRed{$\tm=\sctxp{\l\var.\tmthree}\tmfour$            }{
											$\sctxp{\tmthree\esub\var\tmfour}=\tmtwo$		       }{
											$\unf{\tm}=(\l\var.\unf{\tmthree}\sctximp) (\unf{\tmfour})   $}{
											$\unf{\tmthree}\sctximp \isub\var{\unf{\tmfour}}=\unf{\tmthree}\isub\var{\unf{\tmfour}}\sctximp   =\unf{\tmtwo}$ }{
											$\db$}{$\beta$}{\unfsym}{\unfsym}
		\end{center}
		where the first equality in the South-East corner is given by the fact that $\var$ does not occur in $\sctx$ and the variables on which $\sctximp$ substitutes do not occur in $\tmfour$, as it is easily seen by looking at the starting term. Thus the implicit substitutions $\sctximp$ and $ \isub\var{\unf{\tmfour}}$ commute. The redex $\redex$ is \lou\ and $\unf{\redex}$ is the \lo\ $\beta$-redex in $\unf{\tm}$.

	\item \emph{Abstraction $\ctx=\l\var.\ctxtwo$}. It follows immediately by the \ih.
	
	\item \emph{Left of an application $\ctx=\ctxtwo\tmthree$}. By \reflemma{ctx-unf}.\ref{p:ctx-unf-three} we know that $\unf{\tm}=\unf{\ctxtwo}\ctxholep{\relunf{\tmp}{\ctxtwo}}\unf{\tmthree}$. Using the \ih\ we derive the following diagram:		
		\begin{center}
			\commDiagramRed{$
											\ctxtwop\tmp\tmthree
											$}{$
											\ctxtwop\tmtwop\tmthree
											$}{$
											\unf{\ctxtwo}\ctxholep{\relunf{\tmp}{\ctxtwo}}\unf{\tmthree}
											$}{$
											\unf{\ctxtwo}\ctxholep{\relunf{\tmtwop}{\ctxtwo}}\unf{\tmthree}
											$}{$
											\db$}{$\beta$}{\unfsym}{\unfsym}
		\end{center}
		Suppose that the redex $\redex$ reduced in the top side of the diagram is the \lou\ redex in $\tm$, and so in $\ctxtwop\tmp$. The \ih\ also gives that the  redex $\unf{\redex}$ reduced in the bottom side is \lo\ in $\unf{\ctxtwo}\ctxholep{\relunf{\tmp}{\ctxtwo}}$. Suppose it is not \lo\ in $\unf{\tm}$. This is only possible if $\unf{\ctxtwo}\ctxholep{\relunf{\tmp}{\ctxtwo}}$ is an abstraction so that $\unf{\tm}$ is a $\beta$-redex. Note that $\ctxtwop\tmp$ is not of the form $\sctxp{\l\var.\tmfive}$, otherwise the step $\ctxtwop\tmp\tmthree\todb\ctxtwop\tmtwop\tmthree$ would not be the \lou\ step in $\tm$. Then $\ctxtwop\tmp$ has the form $\sctxp\var$ (because a term of the form $\sctxp{\tmfour\tmfive}$ would not unfold into an abstraction), which is $\todb$-normal, absurd.

\item \emph{Right of an application $\ctx=\tmthree\ctxtwo$}. By \reflemma{ctx-unf}.\ref{p:ctx-unf-three} we know that $\unf{\tm}=\unf{\tmthree}\unf{\ctxtwo}\ctxholep{\relunf{\tmp}{\ctxtwo}}$. Using the \ih\ we derive the following diagram:		
		\begin{center}
			\commDiagramRed{$
											\tmthree\ctxtwop\tmp
											$}{$
											\tmthree\ctxtwop\tmtwop
											$}{$
											\unf{\tmthree}\unf{\ctxtwo}\ctxholep{\relunf{\tmp}{\ctxtwo}}
											$}{$
											\unf{\tmthree}\unf{\ctxtwo}\ctxholep{\relunf{\tmtwop}{\ctxtwo}}
											$}{$
											\db$}{$\beta$}{\unfsym}{\unfsym}
		\end{center}
		Suppose that the redex $\redex$ reduced in the top side of the diagram is the \lou\ redex in $\tm$, and so in $\ctxtwop\tmp$. The \ih\ also gives that the  redex $\unf{\redex}$ reduced in the bottom side is \lo\ in $\unf{\ctxtwo}\ctxholep{\relunf{\tmp}{\ctxtwo}}$. Suppose it is not \lo\ in $\unf{\tm}$. Note that $\tmthree$ is an \opt\ normal form. Then by \refprop{opt-nf-to-nf} $\unf{\tmthree}$ is a $\beta$-normal form. The only possibility is that $\unf{\tmthree}$ is an abstraction so that $\unf{\tm}$ has a root $\beta$-redex. Note that $\tmthree$ is not of the form $\sctxp{\l\var.\tmfive}$, otherwise the step $\tmthree\ctxtwop\tmp\todb\tmthree\ctxtwop\tmtwop$ would not be the \lou\ step in $\tm$. Then $\tmthree$ has the form $\sctxp\var$ (because a term of the form $\sctxp{\tmfour\tmfive}$ would not unfold into an abstraction) and $\sctx$ can act on $\var$ substituting a term $\tmsix$ s.t. $\relunf{\tmsix}{\sctx}$ is an abstraction. Since that occurrence of $\var$ is in an applicative context in $\tm$ we get that $\tm$ has an \opt\ redex on the left of $\redex$, absurd.

	\item \emph{Substitution $\ctx=\ctxtwo\esub\var\tmfour$}. By \ih\ $\unf{\ctxtwo}\ctxholep{\relunf{\tmp}{\ctxtwo}}\tob \unf{\ctxtwo}\ctxholep{\relunf{\tmtwop}{\ctxtwo}}$
	Then:
\[\begin{array}{llllllll}
\unf{\tm}&=&
\unf{\ctxtwo}\ctxholep{\relunf{\tmp}{\ctxtwo}}\isub\var{\unf\tmfour} &= && \reflemma{ctx-unf}.\ref{p:ctx-unf-two}\\
&&\unf{\ctxtwo}\isub\var{\unf\tmfour}\ctxholep{\relunf{\tmp}{\ctxtwo}\isub\var{\unf\tmfour}}
&\tob & &\reflemma{beta-sub-left}\\
&&\unf{\ctxtwo}\isub\var{\unf\tmfour}\ctxholep{\relunf{\tmtwop}{\ctxtwo}\isub\var{\unf\tmfour}} &=&& \reflemma{ctx-unf}.\ref{p:ctx-unf-two}\\
&&\unf{\ctxtwo}\ctxholep{\relunf{\tmtwop}{\ctxtwo}}\isub\var{\unf\tmfour}&=&
\unf{\tmtwo}
\end{array}\]

	Suppose that the redex $\redex$ reduced in $\tm$ is \lo\ but the redex $\unf{\redex}$ reduced in $\unf{\tm}$ is not. The \ih\ also gives that the redex $\unf{\ctxtwo}\ctxholep{\relunf{\tmp}{\ctxtwo}}\tob \unf{\ctxtwo}\ctxholep{\relunf{\tmtwop}{\ctxtwo}}$ is  \lo. Consequently, the \lo\ redex of $\unf{\tm}$ has been created by the substitution $\isub\var{\unf\tmfour}$, \ie\ $\unf{\ctxtwo}\ctxholep{\relunf{\tmp}{\ctxtwo}}=\ctxthreep\var$ with $\ctxthree\leftout\unf{\ctxtwo}$ and ${\unf\tmfour}$ has a $\beta$-redex or $\unf\tmfour$ is an abstraction and $\ctxthree$ is applicative. Before treating the two cases separately, we deal with some common facts. By \reflemma{varoc-relunf-new} there is a context $\ctxfour$ s.t. $\ctxtwop\tmp=\ctxfourp\vartwo$ and $\unf{\ctxfour}\outin\ctxthree$. By $\ctxthree\leftout\unf{\ctxtwo}$ we obtain $\unf{\ctxfour}\leftout\unf{\ctxtwo}$, and by \reflemma{lefttor-prop} $\ctxfour\leftout\ctxtwo$. Then $\ctxfour\esub\var\tmfour\leftout\ctxtwo\esub\var\tmfour=\ctx$. To conclude we need to show that $\ctxfour\esub\var\tmfour$ is the position of an \opt\ redex, that being on the left of $\ctx$ would give us a contradiction. Cases:
	
	\begin{varenumerate}
		\item \emph{$\unf\tmfour$ has a $\beta$-redex}. From $\unf{\ctxtwo}\ctxholep{\relunf{\tmp}{\ctxtwo}}=\ctxthreep\var$, $\ctxtwop\tmp=\ctxfourp\vartwo$ and $\unf{\ctxfourp\vartwo}=_{\reflemma{ctx-unf}.\ref{p:ctx-unf-two}}\unf{\ctxfour}\ctxholep{\relunf{\vartwo}{\ctxfour}}$ it follows that $\unf{\ctxfour}\ctxholep{\relunf{\vartwo}{\ctxfour}}=\ctxthreep\var$, and from $\unf{\ctxfour}\outin\ctxthree$ we obtain $\var\in\fv{\relunf{\vartwo}{\ctxfour}}$. So, $\unf\tmfour$ is a subterm of $\relunf{\vartwo}{\ctxfour\esub\var\tmfour}=\relunf{\vartwo}{\ctxfour}\isub\var{\unf\tmfour}$. Summing up:
		\begin{varenumerate}
			\item $\tm=\ctxfourp{\vartwo}\esub\var\tmfour$,
			\item $\relunf{\vartwo}{\ctxfour\esub\var\tmfour}$ contains a $\beta$-redex.
		\end{varenumerate} 
		
		Then $\ctxfour\esub\var\tmfour$ is the position of a \opt\ redex in $\tm$ on the left of $\redex$, absurd.
		
		\item \emph{$\unf\tmfour$ is an abstraction and $\ctxthree$ is applicative}. By \reflemma{varoc-relunf-new}.\ref{p:varoc-relunf-new-2} there are two sub-cases:
		
		\begin{varenumerate}
			\item \emph{$\ctxfour$ is applicative and $\unf{\ctxfour}=\ctxthree$}. By \refremark{compact-condition-2} we obtain $\relunf{\vartwo}{\ctxfour}=\var$, and so:
			\begin{varenumerate}
				\item $\tm=\ctxfourp{\vartwo}\esub\var\tmfour$,
				
				\item $\ctxfour\esub\var\tmfour$ is an applicative context,
				
				\item $\relunf{\vartwo}{\ctxfour\esub\var\tmfour}=\relunf{\vartwo}{\ctxfour}\isub\var{\unf\tmfour}=\var\isub\var{\unf\tmfour}=\unf\tmfour$ is an abstraction.
			\end{varenumerate}
			Then $\tm$ has a \opt\ redex on the left of $\redex$, absurd.
			
			\item \emph{there exists an applicative context $\apctx$ s.t. $\relunf{\vartwo}{\ctxfour}=\apctxp\var$}. Then:
			\begin{varenumerate}
				\item $\tm=\ctxfourp{\vartwo}\esub\var\tmfour$,
				
				\item $\relunf{\vartwo}{\ctxfour\esub\var\tmfour}=\relunf{\vartwo}{\ctxfour}\isub\var{\unf\tmfour}=\apctxp\var\isub\var{\unf\tmfour}=\apctxp{\unf\tmfour}$ is a $\beta$-redex.
			\end{varenumerate}
			Then $\tm$ has a \opt\ redex on the left of $\redex$, absurd.
		\end{varenumerate}
	\end{varenumerate}
	
\end{varenumerate}
\end{proof}

Projection of derivations now follows as an easy induction:

\begin{theorem}[Projection]
	\label{tm:projection}
	Let $\tm$ be a \lsc\ term and $\deriv:\tm\tolou^*\tmtwo$. Then there is a \lo\ $\beta$-derivation $\unf{\deriv}:\unf{\tm}\tob^*\unf{\tmtwo}$ s.t. $\size{\unf{\deriv}}=\sizedb{\deriv}$.
\end{theorem}

\begin{proof}
	 By induction on the length $k$ of $\deriv$. If $k=0$ the statement is trivially true. If $k>0$ then $\deriv$ has the form $\tm\to^k\tmfour\to\tmtwo$ and the prefix $\tm\to^k\tmfour$ by \ih\ verifies the statement. Cases of $\tmfour\to\tmtwo$:
	\begin{varenumerate}
		\item $\tmfour\tols\tmtwo$: then $\unf{\tmfour}=\unf{\tmtwo}$, and there is nothing to prove. 
			
		\item $\tmfour\todb\tmtwo$: by \reflemma{useful-projection} such a step unfolds to exactly one \lo\ $\beta$-step $\unf{\tmfour}\tob\unf{\tmtwo}$, that together with the \ih\ proves the statement.
	\end{varenumerate}
\end{proof}

\ignore{
%%%%%
%% BEN: da finire per mostrare che useful-normalizable iff shallow-normalizable
%%%
\begin{lemma}
	Let $\tm$ be a useful normal form and consider $\deriv:\tm\to^*\tmtwo$. Then $\tmtwo$ is a useful normal form and $\sizedb\deriv=0$.
	\end{varenumerate}
\end{lemma}

\begin{proof}
	By induction on $k=\size\deriv$. If $k=0$ then the statement trivially holds. If $k>0$ then by \ih\ $\deriv$ has the form $\tm\tols^{k-1}\tmthree\to\tmtwo$ with $\tmthree$ useful normal form and---as such---$\db$-normal. Consequently, the last step $\tmthree\to\tmtwo$ is a $\ls$-step. We are left to show that $\tmthree$ is a useful normal form. By induction on the position $\ctx$ of the last step $\tmthree=\ctxp\var\tols\ctxp\tmfour=\tmtwo$.
	\begin{varenumerate}
		\item \emph{Empty context, \ie\ $\ctx=\ctxhole$}. Impossible, because $\tmthree=\var$ would be normal.
		
		\item \emph{Abstraction, \ie\ $\ctx=\lambda\var.\ctxtwo$}. It follows by the \ih.
		
		\item \emph{Left of an Application, \ie\ $\ctx=\ctxtwo\tmfive$}. 
	\end{varenumerate}
	
\end{proof}
}

% !TEX root = main.tex
\section{The Syntactic Bound Property, via Nested Derivations}
\label{sect:nested}
In this section we show that \lou\ derivations have the syntactic
bound property. Instead of proving this fact directly, we introduce an
abstract property, the notion of \emph{nested derivation} and then
prove that 1) nested derivations ensure the syntactic bound property,
and 2) \lou\ derivations are nested. Such an approach helps to
understand both \lou\ derivations and the syntactic bound property.

\begin{definition}[Nested Derivation]
Two $\ls$-steps $\tm\tols\tmtwo\tols\tmthree$ are \deff{nested} if the second one substitutes on the subterm substituted by the first one, \ie\ if exist $\ctx$ and $\ctxtwo$ s.t. the two steps have the compact form $\ctxp{\var}\tols\ctxp{\ctxtwop{\vartwo}}\tols\ctxp{\ctxtwop{\tmtwo}}$. A derivation is nested if any two consecutive substitution steps are nested.
\end{definition}

For instance, the first of the following two sequences of steps is nested while the second is not:
\begin{align*}
  (\ap\var\vartwo)\esub\var{\ap\vartwo\tm}\esub\vartwo\tmtwo&\tols(\ap{(\ap\vartwo\tm)}\vartwo)\esub\var{\ap\vartwo\tm}\esub\vartwo\tmtwo\\
    &\tols(\ap{(\ap\tmtwo\tm)}\vartwo)\esub\var{\ap\vartwo\tm}\esub\vartwo\tmtwo;\\
  (\ap\var\vartwo)\esub\var{\ap\vartwo\tm}\esub\vartwo\tmtwo&\tols(\ap{(\ap\vartwo\tm)}\vartwo)\esub\var{\ap\vartwo\tm}\esub\vartwo\tmtwo\\
    &\tols(\ap{(\ap\vartwo\tm)}\tmtwo)\esub\var{\ap\vartwo\tm}\esub\vartwo\tmtwo.
\end{align*}
The idea is that nested derivations ensure the syntactic bound property because no substitution can be used twice in a nested sequence $\tmtwo\tols^k\tmthree$, and so $k$ is necessarily bounded by $\esmeas{\tmtwo}$.

%The next lemma shows that nested derivations have the syntactic bound property:

\begin{lemma}[Nested + Subterm = Syntactic Bound]
\label{l:nested-trace}
Let $\tm$ be a $\l$-term, $\deriv:\tm\to^n \tmtwo\tols^k\tmthree$ be a derivation having the subterm property and whose suffix $\tmtwo\tols^k\tmthree$ is nested. Then $k\leq\esmeas{\tmtwo}$.
\end{lemma}

\begin{proof}
Let $\tmtwo=\tmtwo_0\tols\tmtwo_1\tols\ldots\tols \tmtwo_k=\tmthree$ be the nested suffix of $\deriv$ and $\tmtwo_i\tols\tmtwo_{i+1}$ one of its steps, for $i\in\set{0,\ldots,k-2}$. Let us use $\ctx_i$ for the external context of the step, \ie\ the context s.t. $\tmtwo_i=\ctx_i\ctxholep{\ctxtwop{\var}\esub{\var}{\tmfour}}\tols \ctx_i\ctxholep{\ctxtwop{\tmfour}\esub{\var}{\tmfour}}=\tmtwo_{i+1}$. The following nested step $\tmtwo_{i+1}\tols\tmtwo_{i+2}$ substitutes on the substituted occurrence of $\tmfour$. By the subterm property, $\tmfour$ is a subterm of $\tm$ and so it is has no explicit substitution. Then the explicit substitution acting in $\tmtwo_{i+1}\tols\tmtwo_{i+2}$ is on the right of $\esub{\var}{\tmfour}$, \ie\ the external context $\ctx_{i+1}$ is a prefix of $\ctx_i$, in symbols $\ctx_{i+1}\outin\ctx_i$.  Since the derivation $\tmtwo_0\tols\tmtwo_1\tols\ldots\tols \tmtwo_k$ is nested we obtain a sequence $\ctx_k\outin\ctx_{k-1}\outin\ldots\outin\ctx_0$ of contexts of $\tmtwo$. In particular, every $\ctx_i$ corresponds to a different explicit substitution in $\tmtwo$, and so $k\leq\esmeas{\tmtwo}$.
\end{proof}

We are left to show that our small-step implementation of $\beta$ --- \lou\ derivations --- indeed are nested derivations with the subterm property. We already know that they have the subterm property (\refcoro{subterm-for-lou}), so we only need to show that they are nested.

We need a technical lemma. A context $\ctx$ is \lo\ if given a fresh variable $\var$ the position $\ctxtwo$ of any redex in $\ctxp\var$ is s.t. $\ctx\leftout\ctxtwo$. 

\begin{lemma}\label{l:lookunfold}
Let $\tm$ be a $\l$-term in normal form and $\ctx$ a \lo\ context (possibly containing ES). If $\relunf{\tm}{\ctx}$ has a $\beta$-redex then there is a context $\ctxtwo$ s.t. $\tm=\ctxtwop\var$, $\ctx$ acts on $\var$, and the induced step $\ctxp{\ctxtwop{\var}}\tols\ctxp{\ctxtwop{\tmtwo}}$ is \lou.
\end{lemma}

\begin{proof}
By induction on $\tm$. Cases:
\begin{varenumerate}
	\item \emph{Variable $\tm=\var$}. Then $\ctx$ acts on $\var$ otherwise $\relunf{\tm}{\ctx}$ would not have a $\beta$-redex, \ie\ $\ctx$ has the form $\ctxthreep{\ctxfour\esub \var\tmtwo}$. Let $\ctxtwo\defeq\ctxhole$ and consider the induced step $\ctxp\var\tols\ctxp\tmtwo$. To prove that it is useful we have to analyse $\relunf{\tmtwo}{\ctxthree}$. Now, $\relunf{\tm}{\ctx}=\relunf{\var}{\ctx}=\relunf{\var}{\ctxthreep{\ctxfour\esub \var\tmtwo}}=\relunf{\var}{\ctxthreep{\ctxhole\esub \var\tmtwo}}=\relunf{\tmtwo}{\ctxthree}$ and so $\relunf{\tmtwo}{\ctxthree}$ has a $\beta$-redex, that is, the step is \opt. Since $\ctxtwo=\ctxhole$ and $\ctx$ is \lo, it follows that the step is also \lo.

	\item \emph{Abstraction $\tm=\l \vartwo.\tmthree$}. Note that $\relunf{\tm}{\ctx}=\relunf{(\l \vartwo.\tmthree)}{\ctx}=\relunf{\tmthree}{\ctxp{\l \vartwo.\ctxhole}}$. Then $\relunf{\tmthree}{\ctxp{\l \vartwo.\ctxhole}}$ has a $\beta$-redex and we can apply the \ih, and obtain a context $\ctxtwo'$ satisfying the statement wrt $\tmthree$. To conclude it is enough to take $\ctxtwo\defeq\l\vartwo.\ctxtwo'$.

	\item \emph{Application $\tm=\tmthree\tmfour$}. Note that the hypothesis on $\tm$ forbids $\tmthree$ to be an abstraction. There are three sub-cases, depending on where is the \lo\ $\beta$-redex in $\relunf{\tm}{\ctx}=\relunf{\tmthree}{\ctx}\relunf{\tmfour}{\ctx}$:

	\begin{varenumerate}
		\item \emph{It is the outermost application of $\relunf{\tm}{\ctx}$, \ie\ $\relunf{\tmthree}{\ctx}$ is an abstraction}. Since $\tmthree$ is not an abstraction, it can only be a variable $\var$ on which $\ctx$ acts by substituting $\tmtwo$ (\ie\ $\ctx=\ctxthreep{\ctxfour\esub \var\tmtwo}$). As in the variable case, we obtain $\relunf{\var}{\ctx}=\relunf{\tmtwo}{\ctxthree}$, \ie\ $\relunf{\tmtwo}{\ctxthree}$ is an abstraction. Then let $\ctxtwo\defeq \ctxhole\tmfour$. The induced step $\ctxp{\ctxtwop{\var}}\tols\ctxp{\ctxtwop{\tmtwo}}$ is \opt\ because $\ctxtwo\defeq\var\tmfour$ is applicative, and it is \lo\ because both $\ctx$ and $\ctxtwo$ are \lo.
		
		\item \emph{It is in the left subterm $\relunf{\tmthree}{\ctx}$}. It is enough to apply the \ih\ as in the abstraction case (obtaining a context $\ctxtwo'$ satisfying the statement wrt to $\tmthree$, and then take $\ctxtwo\defeq \ctxtwo'\tmfour$).
		\item \emph{It is in the right subterm $\relunf{\tmfour}{\ctx}$}. Analogous to the previous case.
	\end{varenumerate}
\end{varenumerate}
\end{proof}

We conclude with:

\begin{proposition}
\lou\ derivations are nested, and so they have the syntactic bound property.
\end{proposition}

\begin{proof}
We prove the following implication: if the reduction step
$\ctxp{\var}\tols\ctxp{\tmtwo}$ is
\lou, and the \lou\ redex in $\ctxp\tmtwo$ is a $\ls$-redex then its position $\ctxtwo$ occurs in $\tmtwo$, \ie\ $\ctx\outin\ctxtwo$ or $\ctx=\ctxtwo$, \ie\ the two steps are nested. The syntactic bound property then follows from \refprop{lou-der-are-standard} (\lou\ derivations are standard), \refcoro{subterm}.\ref{p:subterm-1} (standard derivations have the subterm property), and \reflemma{nested-trace} (nested plus the subterm property implies the syntactic bound property). Two cases, depending on \emph{why} the reduction
step $\ctxp{\var}\tols\ctxp{\tmtwo}$ is \opt:
\begin{varenumerate}
\item
   \emph{$\ctx$ is applicative and $\relunf{\tmtwo}{\ctx}$ is an
  abstraction}. Two sub-cases:
  \begin{varenumerate}
  \item
    \emph{$\tmtwo$ is an abstraction}. Then the \lou\ redex in $\ctxp{\tmtwo}$ is the multiplicative redex having $\tmtwo$ as abstraction, and there is nothing to prove (because the \lou\ redex is not a substitution redex).
  \item
    \emph{$\tmtwo$ is not an abstraction}. Then it must a variable $\varthree$ (because it is a $\lambda$-term),
    and $\relunf{\varthree}{\ctx}$ is an abstraction. But then $\ctxp{\tmtwo}$
    is simply $\ctxp{\varthree}$ and the given occurrence of $\varthree$ marks another
    \opt\ substitution redex, that is the \lou\ redex because $\ctx$ already was the position of the \lou\ redex at the preceding step.
  \end{varenumerate}
\item
  \emph{$\relunf{\tmtwo}{\ctx}$ is not an abstraction or $\ctx$ is not applicative, but $\relunf{\tmtwo}{\ctx}$ contains a $\beta$-redex}. Two sub-cases:
  \begin{varenumerate}
  \item
    \emph{$\tmtwo$ contains a $\db$-redex itself}. Then consider the position $\ctxthree$ of the \lo\ $\db$-redex $\redex$ in $\tmtwo$. Two sub cases:
    \begin{varenumerate}
    		\item \emph{$\redex$ is also the \lou\ redex in $\ctxp\tm$}. Then there is nothing to prove, because the \lou\ redex is not a substitution redex.
		\item \emph{$\redex$ is not the \lou\ redex in $\ctxp\tm$}. Then there is a \lou\ substitution redex $\redextwo$ of position $\ctxfour\leftout\ctxp\ctxthree$. Since 
		\begin{varenumerate}
			\item $\relunf{\tmtwo}{\ctx}$ is not an abstraction or $\ctx$ is not applicative, 
			\item the previous step was the \lou\ step, 
			\item useless steps do not become useful by reducing useful redexes on their right \reflemma{useless-pers}
		\end{varenumerate}
		we necessarily have $\ctx\leftout\ctxfour$ and we conclude.
    \end{varenumerate}
  \item
    \emph{$\tmtwo$ does not contain a redex}. Remember that $\relunf{\tmtwo}{\ctx}$ does contain a $\beta$-redex, so we can apply \reflemma{lookunfold} and obtain that there exists $\ctxthree$ s.t. $\tmtwo=\ctxthreep{\varthree}$ and 
    $\ctxp{\ctxthreep{\varthree}}$ identifies the \lou\ redex.
  \end{varenumerate}
\end{varenumerate} 
\end{proof}

At this point, we proved all the abstract properties implying the high-level implementation theorem. %The next section proves the selection property for \lou\ derivations, completing the requirements for the low-level implementation theorem, and thus concluding our study.

% !TEX root = main.tex
%%%%%%%%%%%%%%%%%%%%%%%%%%%%%%%%%%%%%%%%%%%%%%%%%%%%%%%%%%%%%%%%%%%%%%%%%%%%%%%%%%%%
\section{The Selection Property, or Computing Functions in Compact Form}\label{sect:properties}
\label{sect:algorithm}
%%%%%%%%%%%%%%%%%%%%%%%%%%%%%%%%%%%%%%%%%%%%%%%%%%%%%%%%%%%%%%%%%%%%%%%%%%%%%%%%%%%%
\newcommand{\alg}{\mathcal{A}}
\newcommand{\algone}{\mathcal{A}}
\newcommand{\algtwo}{\mathcal{B}}
\newcommand{\algthree}{\mathcal{C}}
\newcommand{\fun}{f}
\newcommand{\funone}{f}
\newcommand{\funtwo}{g}
\newcommand{\funequal}{f_{=}}
\newcommand{\nattm}{\mathbb{T}}
\newcommand{\natvar}[1]{\mathsf{var}(#1)}
\newcommand{\natlam}{\mathsf{lam}}
\newcommand{\natapp}{\mathsf{app}}
\newcommand{\natval}{n}
\newcommand{\bool}{\mathbb{B}}
\newcommand{\btrue}{\mathsf{true}}
\newcommand{\bfalse}{\mathsf{false}}
\newcommand{\bval}{b}
\newcommand{\funnat}{\mathit{nature}}
\newcommand{\funred}{\mathit{redex}}
\newcommand{\funav}{\mathit{apvars}}
\newcommand{\funfv}{\mathit{freevars}}
\newcommand{\vars}{\mathcal{VARS}}
\newcommand{\ifnempty}[3]{#1\Downarrow_{#2,#3}}
\newcommand{\len}[1]{|#1|}
\newcommand{\eq}{\mathit{alpha}}
\renewcommand{\setone}{V}
\renewcommand{\settwo}{W}
\newcommand{\targetset}{A}

This section proves the selection property for \lou\ derivations,
which is the missing half of the proof that they are mechanisable,
\ie\ that they enjoy the low-level implementation theorem. The proof
consists in providing a polynomial algorithm for testing the
usefulness of a substitution step. The subtlety is that the test has
to check whether a term in the form $\relunf{\tm}{\ctx}$ contains a
$\beta$-redex, or whether it is an abstraction, without explicitly
computing $\relunf{\tm}{\ctx}$ (which, of course, takes exponential
time in the worst case). If one does not prove that this can be done
in time polynomial in (the size of) $\tm$ and $\ctx$, then firing
\emph{each} reduction step can cause an exponential blowup!

Our algorithm consists in the simultaneous computation of 4 correlated
functions on terms in compact form, two of which will provide the
answer to our problem. We need some abstract preliminaries about
computing functions in compact form.

A function $\funone$ from $n$-uples of $\lambda$-terms to
a set $\targetset$ is said to have \emph{arity} $n$, and we write 
$\funone:n\rightarrow\targetset$ in this case. 
%For example, the function $\funequal$ computing the equality (modulo $\alpha$) 
%of two terms has signature $2\rightarrow\{\btrue,\bfalse\}$. 
The function $\funone$ is said to be:
\begin{varitemize}
\item
  \emph{Efficiently computable} if there is a polynomial time
  algorithm $\alg$ such that for every $n$-uple of $\lambda$-terms
  $(\tm_1,\ldots,\tm_n)$, the result of
  $\alg(\tm_1,\ldots,\tm_n)$ is precisely $\funone(\tm_1,\ldots,\tm_n)$.
\item
  \emph{Efficiently computable in compact form} if there is a polynomial
  time algorithm $\alg$ such that for every $n$-uple of LSC terms
  $(\tm_1,\ldots,\tm_n)$, the result of
  $\alg(\tm_1,\ldots,\tm_n)$ is precisely $\funone(\unf{\tm_1},\ldots,\unf{\tm_n})$.
\item
  \begin{sloppypar}
    \emph{Efficiently computable in compact form relatively to a context} if there is a polynomial
    time algorithm $\alg$ such that for every $n$-uple of pairs of LSC terms and contexts
    $((\tm_1,\ctx_1),\ldots,(\tm_n,\ctx_n))$, the result of
    $\alg((\tm_1,\ctx_1),\ldots,(\tm_n,\ctx_n)))$ is precisely 
    $\funone(\relunf{\tm_1}{\ctx_1},\ldots,\relunf{\tm_n}{\ctx_n})$.
  \end{sloppypar}
\end{varitemize}
An example of function is $\eq:2\rightarrow\bool$, which given two
$\lambda$-terms $\tm$ and $\tmtwo$, returns $\btrue$ if $\tm$ and
$\tmtwo$ are $\alpha$-equivalent and $\bfalse$
otherwise. In~\cite{DBLP:conf/rta/AccattoliL12}, $\eq$ is shown to be
efficiently computable in compact form, via a dynamic programming
algorithm $\algtwo_\eq$ taking in input two LSC terms and computing,
for every pair of their subterms, whether the (unfoldings) are
$\alpha$-equivalent or not. Proceeding bottom-up, as usual in dynamic
programming, allows to avoid the costly task of computing unfoldings
explicitly, which takes exponential time in the worst-case. More
details about $\algtwo_\eq$ can be found
in~\cite{DBLP:conf/rta/AccattoliL12}.

Each one of the functions of our interest take values in one of the following sets:
\begin{align*}
	\vars &= \mbox{ the set of finite sets of variables}\\
  \bool&=\{\btrue,\bfalse\}\\
    \nattm&=\{\natvar{\var}\mid\mbox{ $\var$ is a variable}\}\cup\{\natlam,\natapp\}
\end{align*}
Elements of $\nattm$ represent the \emph{nature} of a term. The functions are:
\begin{varitemize}
\item
  $\funnat:1\rightarrow\nattm$, which returns the nature of the input term;
\item
  $\funred:1\rightarrow\bool$, which returns $\btrue$ if the input term
  contains a redex and $\bfalse$ otherwise;
\item
  $\funav:1\rightarrow\vars$, which returns the set of variables occurring
  in applicative position in the input term;
\item
  $\funfv:1\rightarrow\vars$, which returns the set of free variables occurring
  in the input term.
\end{varitemize}
Note that they all have arity 1 and that showing $\funred$ and $\funnat$ to be \emph{efficiently computable in compact form relatively to a context}
is precisely what is required to prove the efficiency of useful reduction. 

The four functions above can all be proved to be efficiently computable (in the three
meanings). It is convenient to do so by giving an algorithm computing the product function 
$\funnat\times\funred\times\funav\times\funfv:1\rightarrow\nattm\times\bool\times\vars\times\vars$
(which we call $\funtwo$) compositionally, on the structure of the input term, because the four function are correlated (for example, $\tm\tmtwo$ has a redex, \ie\ $\funred(\tm\tmtwo)=\btrue$, if $\tm$ is an abstraction, \ie\ if $\funnat(\tm)=\natlam$). The algorithm computing $\funtwo$ on terms is $\alg_\funtwo$ and is defined in Figure~\ref{fig:explicit}.
\begin{figure*}
\begin{center}
\fbox{
%\scriptsize
\begin{minipage}{.97\textwidth}
\begin{align*}
  \alg_\funtwo(\var)&=(\natvar{\var},\bfalse,\emptyset,\{\var\});\\
  \alg_\funtwo(\l\var.\tm)&=(\natlam,\bval_\tm,\setone_\tm-\{\var\},\settwo_\tm-\{\var\})\\
  &\mbox{where }\alg_\funtwo(\tm)=(\natval_\tm,\bval_\tm,\setone_\tm,\settwo_\tm);\\
  \alg_\funtwo(\tm\tmtwo)&=(\natapp,\bval_\tm\vee\bval_\tmtwo\vee(\natval_\tm=\natlam),\setone_\tm\cup\setone_\tmtwo\cup\{\var\mid\natval_\tm=\natvar{\var}\},\settwo_\tm\cup\settwo_\tmtwo)\\
  &\mbox{where }\alg_\funtwo(\tm)=(\natval_\tm,\bval_\tm,\setone_\tm,\settwo_\tm)\mbox{ and }\alg_\funtwo(\tmtwo)=(\natval_\tmtwo,\bval_\tmtwo,\setone_\tmtwo,\settwo_\tmtwo);
\end{align*}
\end{minipage}}
\end{center}
\caption{Computing $\funtwo$ in explicit form.}\label{fig:explicit}
\end{figure*}
 
 The interesting case in the algorithms for the two compact cases is the one for ES, that makes use of a special notation: given two sets of variables $\setone,\settwo$ and a variable $\var$, $\ifnempty{\setone}{\var}{\settwo}$ is defined to be $\setone$ if $\var\in\settwo$ and the empty set $\emptyset$ otherwise. The algorithm $\algtwo_\funtwo$ computing $\funtwo$ on LSC terms is defined in Figure~\ref{fig:implicit}.
\begin{figure*}
\begin{center}
\fbox{
%\scriptsize
\begin{minipage}{.97\textwidth}
\begin{align*}
  \algtwo_\funtwo(\var)&=(\natvar{\var},\bfalse,\emptyset,\{\var\});\\
  \algtwo_\funtwo(\l\var.\tm)&=(\natlam,\bval_\tm,\setone_\tm-\{\var\},\settwo_\tm-\{\var\})\\
  &\mbox{where }\algtwo_\funtwo(\tm)=(\natval_\tm,\bval_\tm,\setone_\tm,\settwo_\tm);\\
  \algtwo_\funtwo(\tm\tmtwo)&=(\natapp,\bval_\tm\vee\bval_\tmtwo\vee(\natval_\tm=\natlam),\setone_\tm\cup\setone_\tmtwo\cup\{\var\mid\natval_\tm=\natvar{\var}\},\settwo_\tm\cup\settwo_\tmtwo)\\
  &\mbox{where }\algtwo_\funtwo(\tm)=(\natval_\tm,\bval_\tm,\setone_\tm,\settwo_\tm)\mbox{ and }\algtwo_\funtwo(\tmtwo)=(\natval_\tmtwo,\bval_\tmtwo,\setone_\tmtwo,\settwo_\tmtwo);\\    
  \algtwo_\funtwo( \tm\esub\var\tmtwo)&=(\natval,\bval,\setone,\settwo)\\
  &\mbox{where }\algtwo_\funtwo(\tm)=(\natval_\tm,\bval_\tm,\setone_\tm,\settwo_\tm)\mbox{ and }\algtwo_\funtwo(\tmtwo)=(\natval_\tmtwo,\bval_\tmtwo,\setone_\tmtwo,\settwo_\tmtwo)\mbox{ and:}\\
  &\qquad\natval_\tm=\natvar{\var}\Rightarrow\natval=\natval_\tmtwo;\quad\natval_\tm=\natvar{\vartwo}\Rightarrow\natval=\natvar{\vartwo};\\
  &\qquad\natval_\tm=\natlam\Rightarrow\natval=\natlam;\quad\natval_\tm=\natapp\Rightarrow\natval=\natapp;\\
  &\qquad\bval=\bval_\tm\vee(\bval_\tmtwo\wedge\var\in\settwo_\tm)\vee((\natval_\tmtwo=\natlam)\wedge(\var\in\setone_\tmtwo));\\
  &\qquad\setone=(\setone_\tm-\{\var\})\cup\ifnempty{\setone_\tmtwo}{\var}{\settwo_\tm}\cup\;\{\vartwo\mid\natval_\tmtwo=\natvar{\vartwo}\wedge\var\in\setone_\tm\};\\
  &\qquad\settwo=(\settwo_\tm-\{\var\})\cup\ifnempty{\settwo_\tmtwo}{\var}{\settwo_\tm}
\end{align*}
\end{minipage}}
\end{center}
\caption{Computing $\funtwo$ in implicit form.}\label{fig:implicit}
\end{figure*}
The algorithm computing $\funtwo$ on pairs in the form $(\tm,\ctx)$ (where $\tm$ is a LSC term and $\ctx$ is
a shallow context) is defined in Figure~\ref{fig:implicitcontext}.
\begin{figure*}
\begin{center}
\fbox{
%\scriptsize
\begin{minipage}{.97\textwidth}
\begin{align*}
  \algthree_\funtwo(\tm,\ctxhole)&=\algtwo_\funtwo(\tm);\\
  \algthree_\funtwo(\tm,\l\var.\ctx)&=\algthree_\funtwo(\tm,\ctx);\\
  \algthree_\funtwo(\tm,\ctx\tmtwo)&=\algthree_\funtwo(\tm,\ctx);\\
  \algthree_\funtwo(\tm,\tmtwo\ctx)&=\algthree_\funtwo(\tm,\ctx);\\
  \algthree_\funtwo(\tm,\ctx\esub\var\tmtwo)&=(\natval,\bval,\setone,\settwo)\\
  &\mbox{where }\algthree_\funtwo(\tm,\ctx)=(\natval_{\tm,\ctx},\bval_{\tm,\ctx},\setone_{\tm,\ctx},\settwo_{\tm,\ctx})
   \mbox{ and }\algtwo_\funtwo(\tmtwo)=(\natval_\tmtwo,\bval_\tmtwo,\setone_\tmtwo,\settwo_\tmtwo)\mbox{ and:}\\
  &\qquad\natval_\tm=\natvar{\var}\Rightarrow\natval=\natval_\tmtwo;
  \quad\natval_\tm=\natvar{\vartwo}\Rightarrow\natval=\natvar{\vartwo};\\
  &\qquad\natval_\tm=\natlam\Rightarrow\natval=\natlam;
  \quad\natval_\tm=\natapp\Rightarrow\natval=\natapp;\\
  &\qquad\bval=\bval_\tm\vee(\bval_\tmtwo\wedge\var\in\settwo_{\tm,\ctx})\vee((\natval_\tmtwo=\natlam)\wedge(\var\in\setone_\tmtwo));\\
  &\qquad\setone=(\setone_{\tm,\ctx}-\{\var\})\cup\ifnempty{\setone_\tmtwo}{\var}{\settwo_{\tm,\ctx}}
  \cup\;\{\vartwo\mid\natval_\tmtwo=\natvar{\vartwo}\wedge\var\in\setone_\tm\};\\
  &\qquad\settwo=(\settwo_{\tm,\ctx}-\{\var\})\cup\ifnempty{\settwo_\tmtwo}{\var}{\settwo_{\tm,\ctx}}
\end{align*}
\end{minipage}}
\end{center}
\caption{Computing $\funtwo$ in implicit form, relative to a context}\label{fig:implicitcontext}
\end{figure*}

First of all, we need to convince ourselves about the \emph{correctness} of the proposed algorithms: do they really
compute the function $\funtwo$? Actually, the way the algorithms are defined, namely by primitive recursion on the
input terms, helps very much here: a simple induction suffices to prove the following:
\begin{proposition}
  The algorithms $\alg_\funtwo$,$\algtwo_\funtwo$,$\algthree_\funtwo$ are all correct, namely for every
  $\lambda$-term $\tm$, for every term $\tmtwo$ and for every context $\ctx$, it holds that
  $$
  \alg_\funtwo(\tm)=\funtwo(\tm);\qquad\algtwo_\funtwo(\tmtwo)=\funtwo(\unf{\tmtwo});\qquad\algthree_\funtwo(\tmtwo,\ctx)=\funtwo(\relunf{\tmtwo}{\ctx}).
  $$
\end{proposition}

  \begin{proof}
    \begin{varitemize}
    \item
      The equation $\alg_\funtwo(\tm)=\funtwo(\tm)$ can be proved by induction on the structure of $\tm$. Some interesting
      cases:
      \begin{varitemize}
      \item
        If $\tm=\tmtwo\tmthree$, then we know that:
        \begin{align*}
          \alg_\funtwo(\tmtwo\tmthree)&=(\natapp,\bval_\tmtwo\vee\bval_\tmthree\vee(\natval_\tmtwo=\natlam),\setone_\tmtwo\cup\setone_\tmthree\cup\{\var\mid\natval_\tmtwo=\natvar{\var}\},\settwo_\tmtwo\cup\settwo_\tmthree)\\
          &\mbox{where }\alg_\funtwo(\tmtwo)=(\natval_\tmtwo,\bval_\tmtwo,\setone_\tmtwo,\settwo_\tmtwo)\mbox{ and }\alg_\funtwo(\tmthree)=(\natval_\tmthree,\bval_\tmthree,\setone_\tmthree,\settwo_\tmthree);
        \end{align*}
        Now, first of all observe that $\funred(\tm)=\btrue$ if and only if there is a redex in $\tmtwo$ or 
        a redex in $\tmthree$ or if $\tmtwo$ is a $\lambda$-abstraction. Moreover, the variables occurring in
        applicative position in $\tm$ are those occurring in applicative position in either $\tmtwo$ or in
        $\tmthree$ or $\var$, if $\tmtwo$ is $\var$ itself. Similarly, the variables occurring free in $\tm$
        are simply those occurring free in either $\tmtwo$ or in $\tmthree$. The thesis then descends easily
        from the inductive hypothesis.
      \end{varitemize}
    \item
      The equation $\algtwo_\funtwo(\tmtwo)=\funtwo(\unf{\tmtwo})$ can be proved by induction on the structure of $\tmtwo$, using
      the correctness of $\alg$.
    \item
      The equation $\algthree_\funtwo(\tmtwo,\ctx)=\funtwo(\relunf{\tmtwo}{\ctx})$ can be proved by induction on the structure of $\tmtwo$, using
      the correctness of $\alg$.
    \end{varitemize}
    This concludes the proof.
  \end{proof}
The way the algorithms above have been defined also helps while proving that they work in bounded time, e.g., the number of recursive
calls triggered by $\alg_\funtwo(\tm)$ is linear in $\len{\tm}$ and each of them takes polynomial time. As a consequence, we can also easily bound the complexity of the three algorithms at hand.
\begin{proposition}
  The algorithms $\alg_\funtwo$,$\algtwo_\funtwo$,$\algthree_\funtwo$ all work in polynomial time. Thus \lou\ derivations are mechanisable.
\end{proposition}

  \begin{proof}
    The three algorithms are defined by primitive recursion. More specifically:
    \begin{varitemize}
    \item
      Any call $\alg_\funtwo(\tm)$ triggers at most $\len{\tm}$ calls to $\alg_\funtwo$;
    \item
      Any call $\algtwo_\funtwo(\tm)$ triggers at most $\len{\tm}$ calls to $\algtwo_\funtwo$;
    \item
      Any call $\algthree_\funtwo(\tm,\ctx)$ triggers at most $\len{\tm}+\len{\ctx}$ calls to $\algtwo$ and
      at most $\len{\ctx}$ calls to $\algthree$;
    \end{varitemize}
    Now, the amount of work involved in any single call (not counting the, possibly recursive, calls) 
    is itself polynomial, simply because the tuples produced in output are made of objects whose
    size is itself bounded by the length of the involved terms and contexts.
\end{proof}

% !TEX root = main.tex
%%%%%%%%%%%%%%%%%%%%%%%%
\section{Summing Up}
%%%%%%%%%%%%%%%%%%%%%%%%

The various ingredients from the previous sections can be combined
together so as to obtain the following result:
\begin{theorem}[Invariance]\label{theo:invariance}
  There is an algorithm which takes in input a $\l$-term $\tm$ and
  which, in time polynomial in $\nos{\toblo}{\tm}$ and $\size{\tm}$,
  outputs an \lsc\ term $\tmtwo$ such that $\unf{\tmtwo}$ is the normal
  form of $\tm$.
\end{theorem}
As we have already mentioned, the algorithm witnessing the invariance
of $\l$-calculus does \emph{not} produce in output a $\l$-term, but a
compact representation in the form of a term with
ES. Theorem~\ref{theo:invariance}, together with the fact that
equality of terms can be checked efficiently \emph{in compact form}
entail the following formulation of invariance, akin in spirit to,
\eg, Statman's Theorem~\cite{DBLP:journals/tcs/Statman79a}:
\begin{corollary}\label{coro:statmanlike}
  There is an algorithm which takes in input two $\l$-terms $\tm$ and
  $\tmtwo$ and checks whether $\tm$ and $\tmtwo$ have the same
  normal form in time polynomial in $\nos{\toblo}{\tm}$,
  $\nos{\toblo}{\tmtwo}$, $\size{\tm}$, and $\size{\tmtwo}$.
\end{corollary}

\noindent If one instantiates Corollary~\ref{coro:statmanlike} to the
case in which $\tmtwo$ is a normal form, one obtains that checking
whether the normal form of any term $\tm$ is equal to $\tmtwo$ can be
done in time polynomial in $\nos{\toblo}{\tm}$, $\size{\tm}$, and
$\size{\tmtwo}$. This is particularly relevant when the size of
$\tmtwo$ is constant, \eg, when the $\l$-calculus computes decision
problems and the relevant results are truth values.

Please observe that whenever one (or both) of the involved terms are
\emph{not} normalisable, the algorithms above (correctly) diverge.

% !TEX root = main.tex
%%%%%%%%%%%%%%%%%%%%%%%
\section{Discussion}
%%%%%%%%%%%%%%%%%%%%%%%
Here we further discuss invariance and some potential optimisations,
that, however, are outside the scope of this work (which only deals
with asymptotical bounds and is thus foundational in spirit).

\paragraph{Mechanisability vs Efficiency.} 
Let us stress that the study of invariance is about
\emph{mechanisability} rather than \emph{efficiency}. One is not
looking for the smartest or shortest evaluation strategy. But rather, for
one that does not hide the complexity of its implementation in the
cleverness of its definition, as it is the case for \levy's optimal
evaluation. Indeed, an optimal derivation can be even shorter then the
shortest sequential strategy, but --- as shown by Asperti and Mairson
\cite{DBLP:conf/popl/AspertiM98} --- its definition hides
hyper-exponential computations, so that optimal derivations do not
provide an invariant cost model. The leftmost-outermost strategy, is a
sort of \emph{maximally unshared} normalising strategy, where redexes
are duplicated whenever possible (and unneeded redexes are never
reduced), somehow dually with respect to optimal derivations. It is
exactly this \emph{inefficiency} that induces the subterm property,
the key point for its mechanisability. It is important to not confuse
two different levels of sharing: our \lou\ derivations share
\emph{subterms}, but not \emph{computations}, while \levy's optimal
derivations do the opposite. By sharing computations, they collapse
the complexity of many steps into a single one, making the number of
steps an unreliable measure.

\paragraph{Call-by-Value and Call-by-Need.}
Call-by-name evaluation is in many cases less efficient than
call-by-value or call-by-need evaluation. Since we follow the
call-by-name policy, the same kind of inefficiency shows up
here. However, as already said, invariance is not about absolute
efficiency: call-by-name and call-by-value are incomparable ---
sometimes one can even be exponentially faster than the other,
sometimes the other way around --- but this fact does not forbid both
to be invariant, \ie\ reasonably mechanisable.

We did not prove call-by-value/need invariance. Nonetheless, we strove
to provide an abstract view of both the problem and of the
architecture of our solution, having already in mind the adaptation to
call-by-value/need $\l$-calculi. Recently, the first author and
Sacerdoti Coen show~\cite{valueVariablesDraft} that (in the much
simpler weak case) these policies provide an improved high-level
implementation theorem, where evaluation in the \lsc\ has a
\emph{linear} overhead, rather than quadratic.

\paragraph{Usefulness.} Another source of inefficiency is the fact that at each
reduction step we need to check whether the \lo\ redex is useful
before firing it, and this potentially amounts to doing a global
analysis of the term. One could imagine decorating terms with
additional tags in such a way that the check for usefulness becomes
local \emph{and} updating tags is not too costly, so that useful
reduction may be implemented more efficiently. In particular, building
on the already established relationships between the \lsc\ and
abstract machines \cite{disttilingDraft}, we expect to be able to
design an abstract machine implementing \lou\ evaluation and testing
for usefulness in time linear in the size of the starting term.

% !TEX root = main.tex
%%%%%%%%%%%%%%%%%%%%%%%
\section{Conclusions}
%%%%%%%%%%%%%%%%%%%%%%%
This work can be seen as the last tale in the long quest for an
invariant cost model for the $\lambda$-calculus.  In the last ten
years, the authors have been involved in various works in which
\emph{parsimonious} time cost models have been shown to be invariant
for more and more general notions of reduction, progressively relaxing
the conditions on the use of
sharing~\cite{DBLP:journals/tcs/LagoM08,DBLP:journals/corr/abs-1208-0515,DBLP:conf/rta/AccattoliL12}. None
of the results in the literature, however, concerns reduction to
normal form as instead we do here.

By means of explicit substitutions --- our tool for sharing --- we
provided the first full answer to a long-standing open problem: we
proved that the $\l$-calculus is indeed a reasonable machine, by
showing that the length of the leftmost-outermost derivation to normal
form is an invariant cost model.

The solution required the development of a whole new toolbox: an
abstract deconstruction of the problem, a detailed study of
unfoldings, a theory of useful derivations, and a general view of
functions efficiently computable in compact form. Along the way, we
showed that standard derivations for explicit substitutions enjoy the
crucial \emph{subterm property}. Essentially, it ensures that standard
derivations are mechanisable, unveiling a very abstract notion of
machine hidden deep inside the $\l$-calculus itself, and also a
surprising perspective on the standardisation theorem, a classic
result apparently unrelated to the complexity of evaluation.

Among the downfalls of our results, one can of course mention that
proving systems to characterise time complexity classes equal or
larger than $\mathbf{P}$ can now be done merely by deriving bounds on
the \emph{number} of leftmost-outermost reduction steps to normal
form. This could be useful, e.g., in the context of \emph{light
  logics}~\cite{DBLP:conf/csl/GaboardiR07,DBLP:journals/lmcs/CoppolaLR08,DBLP:journals/iandc/BaillotT09}. The
kind of bounds we obtain here are however more \emph{general} than
those obtained in implicit computational complexity (since we deal
with a universal model of computation). \ignore{Moreover, the emphasis
  here is of course on \emph{relative} bounds.}

While there is room for finer analyses (\eg\ studying call-by-value or
call-by-need evaluation), we consider the understanding of time
invariance essentially achieved. However, the study of complexity
measures for $\l$-terms is far from being over. Indeed, the study of
space complexity for functional programs has only made its very first
steps
\cite{DBLP:conf/lics/Schopp07,DBLP:conf/popl/GaboardiMR08,DBLP:conf/esop/LagoS10},
and not much is known about invariant \emph{space} cost models.

\bibliographystyle{alpha}
\bibliography{\macrospath/biblio}

\begin{thebibliography}{CDLRDR08}

\bibitem[ABKL14]{non-standard-preprint}
Beniamino Accattoli, Eduardo Bonelli, Delia Kesner, and Carlos Lombardi.
\newblock A {N}onstandard {S}tandardization {T}heorem.
\newblock In {\em POPL}, pages 659--670, 2014.

\bibitem[ABM14]{disttilingDraft}
Beniamino Accattoli, Pablo Barenbaum, and Damiano Mazza.
\newblock Distilling {A}bstract {M}achines.
\newblock Accepted to {ICFP} 2014, 2014.

\bibitem[Acc12]{DBLP:conf/rta/Accattoli12}
Beniamino Accattoli.
\newblock An abstract factorization theorem for explicit substitutions.
\newblock In {\em RTA}, pages 6--21, 2012.

\bibitem[ADL12]{DBLP:conf/rta/AccattoliL12}
Beniamino Accattoli and Ugo Dal~Lago.
\newblock On the invariance of the unitary cost model for head reduction.
\newblock In {\em RTA}, pages 22--37, 2012.

\bibitem[AM98]{DBLP:conf/popl/AspertiM98}
Andrea Asperti and Harry~G. Mairson.
\newblock Parallel beta reduction is not elementary recursive.
\newblock In {\em POPL}, pages 303--315, 1998.

\bibitem[AM10]{DBLP:conf/rta/AvanziniM10}
Martin Avanzini and Georg Moser.
\newblock Closing the gap between runtime complexity and polytime
  computability.
\newblock In {\em RTA}, pages 33--48, 2010.

\bibitem[ASC14]{valueVariablesDraft}
Beniamino Accattoli and Claudio Sacerdoti~Coen.
\newblock On the {V}alue of {V}ariables.
\newblock Accepted to {W}o{LLIC} 2014, 2014.

\bibitem[Bar84]{Barendregt84}
Hendrik~Pieter Barendregt.
\newblock {\em The Lambda Calculus -- Its Syntax and Semantics}, volume 103.
\newblock North-Holland, 1984.

\bibitem[BT09]{DBLP:journals/iandc/BaillotT09}
Patrick Baillot and Kazushige Terui.
\newblock Light types for polynomial time computation in lambda calculus.
\newblock {\em Inf. Comput.}, 207(1):41--62, 2009.

\bibitem[CDLRDR08]{DBLP:journals/lmcs/CoppolaLR08}
Paolo Coppola, Ugo Dal~Lago, and Simona Ronchi Della~Rocca.
\newblock Light logics and the call-by-value lambda calculus.
\newblock {\em Logical Methods in Computer Science}, 4(4), 2008.

\bibitem[CF58]{curry1958combinatory}
H.B. Curry and R.~Feys.
\newblock {\em Combinatory Logic}.
\newblock Studies in logic and the foundations of mathematics. North-Holland
  Publishing Company, 1958.

\bibitem[dB87]{deBruijn87}
Nicolaas~G. de~Bruijn.
\newblock {G}eneralizing {A}utomath by {M}eans of a {L}ambda-{T}yped {L}ambda
  {C}alculus.
\newblock In {\em Mathematical Logic and Theoretical Computer Science}, number
  106 in Lecture Notes in Pure and Applied Mathematics, pages 71--92. Marcel
  Dekker, 1987.

\bibitem[DLM08]{DBLP:journals/tcs/LagoM08}
Ugo Dal~Lago and Simone Martini.
\newblock The weak lambda calculus as a reasonable machine.
\newblock {\em Theor. Comput. Sci.}, 398(1-3):32--50, 2008.

\bibitem[DLM12]{DBLP:journals/corr/abs-1208-0515}
Ugo Dal~Lago and Simone Martini.
\newblock On constructor rewrite systems and the lambda calculus.
\newblock {\em Logical Methods in Computer Science}, 8(3), 2012.

\bibitem[DLS10]{DBLP:conf/esop/LagoS10}
Ugo Dal~Lago and Ulrich Sch{\"o}pp.
\newblock Functional programming in sublinear space.
\newblock In {\em ESOP}, pages 205--225, 2010.

\bibitem[GMRDR08]{DBLP:conf/popl/GaboardiMR08}
Marco Gaboardi, Jean-Yves Marion, and Simona Ronchi Della~Rocca.
\newblock A logical account of {PSPACE}.
\newblock In {\em POPL}, pages 121--131, 2008.

\bibitem[GRDR07]{DBLP:conf/csl/GaboardiR07}
Marco Gaboardi and Simona Ronchi Della~Rocca.
\newblock A soft type assignment system for $\lambda$-calculus.
\newblock In {\em CSL}, pages 253--267, 2007.

\bibitem[L{\'e}v78]{thesislevy}
Jean-Jacques L{\'e}vy.
\newblock R\'eductions correctes et optimales dans le lambda-calcul.
\newblock Th\'ese d'Etat, Univ. Paris VII, France, 1978.

\bibitem[LM96]{DBLP:conf/icfp/LawallM96}
Julia~L. Lawall and Harry~G. Mairson.
\newblock Optimality and inefficiency: What isn't a cost model of the lambda
  calculus?
\newblock In {\em ICFP}, pages 92--101, 1996.

\bibitem[Mil07]{DBLP:journals/entcs/Milner07}
Robin Milner.
\newblock Local bigraphs and confluence: Two conjectures.
\newblock {\em Electr. Notes Theor. Comput. Sci.}, 175(3):65--73, 2007.

\bibitem[Ned92]{Ned92}
Robert.~P. Nederpelt.
\newblock The fine-structure of lambda calculus.
\newblock Technical Report CSN 92/07, Eindhoven Univ. of Technology, 1992.

\bibitem[PJ87]{Pey:ImplFunProgLang:87}
Simon Peyton~Jones.
\newblock {\em The Implementation of Functional Programming Languages}.
\newblock International Series in Computer Science. Prentice-Hall, 1987.

\bibitem[Sch07]{DBLP:conf/lics/Schopp07}
Ulrich Sch{\"o}pp.
\newblock Stratified {B}ounded {A}ffine {L}ogic for {L}ogarithmic {S}pace.
\newblock In {\em LICS}, pages 411--420, 2007.

\bibitem[Sta79]{DBLP:journals/tcs/Statman79a}
Richard Statman.
\newblock The typed lambda-calculus is not elementary recursive.
\newblock {\em Theor. Comput. Sci.}, 9:73--81, 1979.

\bibitem[SvEB84]{DBLP:conf/stoc/SlotB84}
Cees~F. Slot and Peter van Emde~Boas.
\newblock On tape versus core; an application of space efficient perfect hash
  functions to the invariance of space.
\newblock In {\em STOC}, pages 391--400, 1984.

\bibitem[Wad71]{Wad:SemPra:71}
C.~P. Wadsworth.
\newblock {\em Semantics and pragmatics of the lambda-calculus}.
\newblock {Ph{D} Thesis}, Oxford, 1971.
\newblock Chapter 4.

\end{thebibliography}

\end{document}